\documentclass[11pt]{article}

\usepackage[margin=1in]{geometry}      
\usepackage{setspace}                   
\usepackage{amsmath}                    
\usepackage{amssymb}                    
\usepackage{amsthm}                     
\usepackage[ruled,vlined,linesnumbered]{algorithm2e} 
\usepackage{tikz}                       
\usepackage{graphicx}                   
\usepackage{cite}                       
\usepackage{float}
\usepackage{hyperref}                   
\SetArgSty{textnormal}

\definecolor{cutA}{RGB}{214,39,40}
\definecolor{cutB}{RGB}{31,119,180}
\definecolor{cutC}{RGB}{44,160,44}
\definecolor{cutD}{RGB}{23,190,207}
\definecolor{cutE}{RGB}{148,103,189}
\definecolor{cutF}{RGB}{140,86,75}

\newcommand{\br}[1]{\mathopen{}\left( #1 \right)}
\newcommand{\brc}[1]{\mathopen{}\left\{ #1 \right\}}
\newcommand{\spr}[1]{\mathopen{}\left| #1 \right|}
\newcommand{\fl}[1]{\mathopen{}\left\lfloor #1 \right\rfloor}

\newcommand{\OPT}{\text{OPT}}
\newcommand{\COST}{{c}}
\newcommand{\dist}{\text{dist}}

\newcommand{\cC}{\mathcal{C}}
\newcommand{\cD}{\mathcal{D}}

\newcommand{\cF}{\mathcal{F}}

\newcommand{\cH}{\mathcal{H}}

\newcommand{\cL}{\mathcal{L}}

\newcommand{\cR}{\mathcal{R}}
\newcommand{\cS}{\mathcal{S}}
\newcommand{\cT}{\mathcal{T}}

\newcommand{\bigo}{\mathcal{O}}

\newtheorem{theorem}{Theorem}[section]
\newtheorem{lemma}[theorem]{Lemma}
\newtheorem{corollary}[theorem]{Corollary}
\theoremstyle{definition}
\newtheorem{observation}[theorem]{Observation}

\title{Hierarchical $\cF$-Clustering: Approximation and Hardness of Clustering into Trees and Bounded Diameter Graphs}
\author{Michał Szyfelbein \\ Gdańsk University of Technology \and Dariusz Dereniowski \\ Gdańsk University of Technology}
\date{}

\begin{document}

\maketitle

\begin{abstract}
    Consider the following variation on the classical \textsc{Hierarchical Clustering} problem: Usually, while building a hierarchical clustering, one recursively partitions the data until each cluster becomes a singleton. In our setup, we relax the halting condition of the recursive process to stop whenever the remaining cluster is a graph belonging to a given class $\mathcal{F}$. We call this problem \textsc{Hierarchical $\mathcal{F}$-Clustering} and we measure the quality of any solution using adapted Dasgupta's clustering objective. We study two natural choices of $\mathcal{F}$: trees, denoted by $\mathcal{T}$ and graphs of diameter bounded by~$d$, denoted by $\mathcal{D}_d$.

    Hereby, we present the first polynomial time $\mathcal{O}(\log n\cdot\log\log n)$ and $\mathcal{O}(\log n)$-approximation algorithms for clustering into $\mathcal{T}$- and $\mathcal{D}_d$-clusters respectively.
    Our main technical contribution is a framework for approximating such problems based on linear programming. In fact, we characterize graphs classes $\mathcal{F}$ for which our approach can be applied and show that it includes both trees and bounded diameter graphs. However, our ideas are not limited to them and might be useful for other structures as well. Broadly speaking, our framework applies whenever the corresponding flat clustering problem, which we call \textsc{$p_{\mathcal{F}}$-Partitioning}, admits a natural ILP formulation together with a rounding procedure with provable approximation guarantees. Intuitively, given a set of vertices called terminals, the problem is to find an edge set whose removal results in satisfying certain vertex-dependent structural predicate for each terminal. For example, excluding all cycles going through terminals or ensuring that the diameter of all clusters with respect to the terminals is not too large. We then use these ingredients to build clustering trees with the aforementioned approximation guarantees. To complement these results, we show that both \textsc{Hierarchical $\mathcal{T}$-Clustering} and \textsc{Hierarchical $\mathcal{D}_d$-Clustering} cannot be approximated within any constant factor under the \textsc{Small Set Expansion Hypothesis}.
\end{abstract}

\noindent
\textbf{Keywords:} Hierarchical $\cF$-Clustering, Feedback Edge Set, Min-Multicut, $\rho$-Separator, Approximation Algorithms, Linear Programming.

\bigskip

\section{Introduction}
Hierarchical clustering is a way to recursively partition a dataset into clusters. The data is usually represented by a weighted, connected graph $G = (V, E, c)$, where $c\colon E\to \mathbb{N}$ is the \emph{similarity} measure between the points. The goal is to find a rooted tree called \textit{clustering tree} $T$, in which every node $u\in V\br{T}$ represents a connected component $G_u\subseteq G$ called a \textit{cluster}. Moreover, every non-leaf node $u$ of $T$ corresponds to a subset of edges $S_u$ such that for every child node $v$ of $u$, $G_v\in G_u\setminus S_u$. For each leaf node $l\in V\br{T}$ it is required that $G_l$ is a singleton, i.~e., $\spr{V\br{G_l}}=1$. Then, the \textit{Dasgupta's clustering objective} is defined as $\COST_G\br{T} = \sum_{u\in V\br{T}} c(S_u) \cdot |V(G_u)|$\footnote{Note that the usual way of defining the Dasgupta's clustering objective is $\COST_G\br{T} = \sum_{uv\in E} c(uv) \cdot |\ell_T\br{u,v}|$, where $\ell_T\br{u,v}$ denotes the smallest subtree containing both $u$ and $v$. It is easy to see, that our formulation is equivalent and we use it for the sake of convenience.}. The problem of finding a clustering tree which minimizes the above cost is known as the \textsc{Hierarchical Clustering} problem and is NP-hard~\cite{ACostFunctionForSimilarityBasedHierarchicalClustering}. Intuitively, the Dasgupta's clustering objective `punishes' cutting edges with high similarity early on in the clustering tree, since the larger the component in which the edge is cut, the higher the contribution of this edge to the cost.

The main advantage of the hierarchical clustering over various kinds of flat clustering is that it provides a decomposition of the data which is not reliant on the granularity of subdivision one wishes to acquire. Since the clusters are refined until they become singletons, the hierarchy encodes various granularity levels into it. However, it is often implicitly assumed that eventually all of them need to become singletons. In this work we ask the following question: What would happen if we relax the halting condition of the recursive partitioning process and allow it to stop whenever the remaining cluster is a graph of some class $\cF$? 
We will assume, that the graph class $\cF$ is a part of the input and represents some kind of structural property which is deemed `simple' enough so that there is no need to further refine the clusters\footnote{Here, we assume that $\cF$ is non-empty and closed under taking subgraphs, which ensures that any subcluster of a cluster in $\cF$ also belongs to $\cF$. Moreover, for the sake of practicality it is reasonable to assume that the class $\cF$ is recognizable in polynomial time. Otherwise, the task of constructing any clustering tree at all becomes intractable and the pursuit of approximation algorithms becomes futile as well.}.
Therefore we only require that for each leaf node $l\in V\br{T}$, $G_l\in \cF$. We call this problem \textsc{Hierarchical $\cF$-Clustering} (HC-$\cF$). The goal is again to minimize $\COST_G\br{T} = \sum_{u\in V\br{T}} c(S_u) \cdot |V(G_u)|$. For a visual example of a clustering tree for \textsc{HC-$\cT$}, where $\cT$ denotes the class of trees see Figure~\ref{fig:hac-intro-example}.
We note that the objective value for HC-$\cF$ may be quite different from that of the classical hierarchical clustering.
If the input graph $G$ belongs to $\cF$, then $\COST_G\br{T}=0$ for a single-node $T$ with root representing $G$.
On the other hand, $\COST_G\br{T}=\Omega\br{n\cdot \log n}$ for any $T$ whose leaves are required to be singletons (assuming uniform costs).
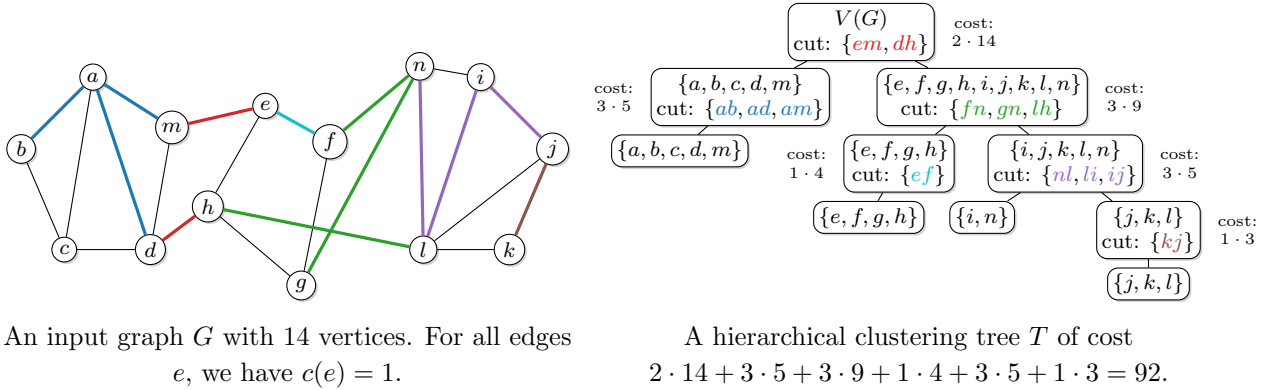
\begin{figure}[H]
    \centering
    \begin{minipage}[t]{0.46\textwidth}
        \centering
        \begin{tikzpicture}[scale=0.95, every node/.style={circle,draw,fill=white,inner sep=1.5pt,font=\scriptsize,preaction={fill=black,opacity=0.18,transform canvas={xshift=0.8pt,yshift=-0.8pt}}}]
            \node (a1) at (-2.6,1.8) {$a$};
            \node (a2) at (-3.6,0.8) {$b$};
            \node (a3) at (-3.0,-0.6) {$c$};
            \node (a4) at (-1.8,-0.6) {$d$};
            \draw[cutB,very thick] (a1)--(a2);
            \draw (a2)--(a3);
            \draw (a3)--(a4);
            \draw[cutB,very thick] (a4)--(a1);
            \draw (a1)--(a3);

            \node (b1) at (-0.2,1.4) {$e$};
            \node (b2) at (0.7,0.9) {$f$};
            \node (b3) at (0.3,-1.1) {$g$};
            \node (b4) at (-1.0,0.0) {$h$};
            \draw[cutD,very thick] (b1)--(b2);
            \draw (b2)--(b3);
            \draw (b3)--(b4);
            \draw (b4)--(b1);

            \node (c1) at (2.8,1.8) {$i$};
            \node (c2) at (3.8,0.8) {$j$};
            \node (c3) at (3.2,-0.6) {$k$};
            \node (c4) at (2.0,-0.6) {$l$};
            \draw[cutE,very thick] (c1)--(c2);
            \draw[cutF,very thick] (c2)--(c3);
            \draw (c3)--(c4);
            \draw[cutE,very thick] (c4)--(c1);
            \draw (c2)--(c4);

            \node (x1) at (-1.5,1.1) {$m$};
            \node (x2) at (1.95,1.95) {$n$};
            \draw[cutB,very thick] (x1)--(a1);
            \draw (x1)--(a4);
            \draw[cutA,very thick] (x1)--(b1);
            \draw[cutC,very thick] (x2)--(b2);
            \draw[cutC,very thick] (x2)--(b3);
            \draw (x2)--(c1);
            \draw[cutE,very thick] (x2)--(c4);
            \draw[cutA,very thick] (a4)--(b4);
            \draw[cutC,very thick] (b4)--(c4);
        \end{tikzpicture}
        \\\smallskip
        {\small An input graph $G$ with 14 vertices.
        For all edges $e$, we have $c(e)=1$.}
    \end{minipage}
    \begin{minipage}[t]{0.52\textwidth}
        \centering
        \begin{tikzpicture}[x=0.85cm,y=0.72cm, hcnode/.style={draw,rounded corners,fill=white,inner sep=2pt,font=\scriptsize,align=center,preaction={fill=black,opacity=0.18,transform canvas={xshift=0.8pt,yshift=-0.8pt}}}]
            \node[hcnode] (r) at (0,0) {$V(G)$\\cut: $\{\textcolor{cutA}{em},\textcolor{cutA}{dh}\}$};
            \node[hcnode,anchor=north] (u1) at ([xshift=-1.60cm,yshift=-0.09cm]r.south) {$\{a,b,c,d,m\}$\\cut: $\{\textcolor{cutB}{ab},\textcolor{cutB}{ad},\textcolor{cutB}{am}\}$};
            \node[hcnode,anchor=north] (u2) at ([xshift=1.60cm,yshift=-0.09cm]r.south) {$\{e,f,g,h,i,j,k,l,n\}$\\cut: $\{\textcolor{cutC}{fn},\textcolor{cutC}{gn},\textcolor{cutC}{lh}\}$};
            \node[hcnode,anchor=north] (v1) at ([xshift=-0.80cm,yshift=-0.11cm]u1.south) {$\{a,b,c,d,m\}$};
            \node[hcnode,anchor=north] (w1) at ([xshift=-1.10cm,yshift=-0.11cm]u2.south) {$\{e,f,g,h\}$\\cut: $\{\textcolor{cutD}{ef}\}$};
            \node[hcnode,anchor=north] (w2) at ([xshift=1.10cm,yshift=-0.11cm]u2.south) {$\{i,j,k,l,n\}$\\cut: $\{\textcolor{cutE}{nl},\textcolor{cutE}{li},\textcolor{cutE}{ij}\}$};
            \node[hcnode,anchor=north] (z1) at ([xshift=-0.40cm,yshift=-0.11cm]w1.south) {$\{e,f,g,h\}$};
            \node[hcnode,anchor=north] (z2) at ([xshift=-1.10cm,yshift=-0.11cm]w2.south) {$\{i,n\}$};
            \node[hcnode,anchor=north] (z3) at ([xshift=1.10cm,yshift=-0.11cm]w2.south) {$\{j,k,l\}$\\cut: $\{\textcolor{cutF}{kj}\}$};
            \node[hcnode,anchor=north] (t1) at ([xshift=0.00cm,yshift=-0.11cm]z3.south) {$\{j,k,l\}$};

            \node[draw=none,fill=none,preaction={},anchor=west,xshift=0.09cm,font=\tiny,align=center] at (r.east) {cost:\\$2\cdot 14$};
            \node[draw=none,fill=none,preaction={},anchor=east,xshift=-0.09cm,font=\tiny,align=center] at (u1.west) {cost:\\$3\cdot 5$};
            \node[draw=none,fill=none,preaction={},anchor=west,xshift=0.09cm,font=\tiny,align=center] at (u2.east) {cost:\\$3\cdot 9$};
            \node[draw=none,fill=none,preaction={},anchor=east,xshift=-0.09cm,font=\tiny,align=center] at (w1.west) {cost:\\$1\cdot 4$};
            \node[draw=none,fill=none,preaction={},anchor=west,xshift=0.09cm,font=\tiny,align=center] at (w2.east) {cost:\\$3\cdot 5$};
            \node[draw=none,fill=none,preaction={},anchor=west,xshift=0.09cm,font=\tiny,align=center] at (z3.east) {cost:\\$1\cdot 3$};

            \draw (r)--(u1);
            \draw (r)--(u2);
            \draw (u1)--(v1);
            \draw (u2)--(w1);
            \draw (u2)--(w2);
            \draw (w1)--(z1);
            \draw (w2)--(z2);
            \draw (w2)--(z3);
            \draw (z3)--(t1);
        \end{tikzpicture}
        \\\smallskip
        {\small A hierarchical clustering tree $T$ of cost $2\cdot14+3\cdot5+3\cdot9+1\cdot4+3\cdot5+1\cdot3=92$.}
    \end{minipage}
    \caption{Example of an instance of \textsc{HC-$\cT$} and a feasible hierarchical clustering: Colors denote edges cut at different levels of the tree.}
    \label{fig:hac-intro-example}
\end{figure}

\subsection{Motivations and applications}
Hierarchical clustering is a clustering technique that is far richer in structure and applications than traditional flat clustering. Since such an object encodes a hierarchical decomposition of data, one may choose an appropriate granularity level of clustering according to their needs. 
Phylogenetic trees, file systems, company management hierarchies, and socio-political divisions are all examples of hierarchical clusterings. Relaxing the halting condition of the clustering process has intuitive applications both for the case of trees as well as bounded-diameter graphs. It is easy to see, that not all of the aforementioned hierarchies finish on singleton objects. For example file systems often contain directories with multiple files. To seek more such examples, one may consider a categorization of products available in a given online store, where one bottom category may contain a very large amount of objects. It seems desirable to find a way to systematically build such structures, especially since a well-structured e-commerce platform is crucial for the success of the business. Here, it would be required for the bottom clusters to consists of products which are similar enough to be considered as a single category.
Enforcing bounded-diameter clusters is therefore a natural requirement, since such clusters aggregate points that are close to one another. Note that in our setting, the cost function encodes similarity between objects, which inflates the diameter of the clusters when the similarity is higher. Luckily, it turns out that our diameter restrictions can be defined with respect to any metric, not limited to the similarity function or the number of edges in the original graph. 

Among possible usages of clustering into trees is the following heuristic for the standard \textsc{HC}. Firstly, run an algorithm for \textsc{HC-$\cT$} to obtain a clustering tree $T$ such that for each leaf $l$ of $T$, $G_l$ is a tree. Next, for each such $G_l$, run any algorithm for \textsc{HC} on trees to recover the remaining levels of clusterization. Since \textsc{HC} on trees admits a polynomial-time constant-factor approximation~\cite{ApproximateHierarchicalClusteringViaSparsestCutAndSpreadingMetrics}, the dominant part of the total approximation error comes from computing the solution for \textsc{HC-$\cT$}. Consequently, the better the approximation algorithm for \textsc{HC-$\cT$}, the better the performance of this heuristic.
The resulting clustering tree can also be used partially dynamically: the $\cT$-clustering tree may be stored statically, and whenever a leaf cluster needs refinement, we can run an \textsc{HC} algorithm on that leaf. In this way, we obtain a partially dynamic clustering tree which occupies less space than a full \textsc{HC} tree, while still allowing on-demand refinement.

%

\subsection{Our results and techniques}

Our main contribution is a general framework for approximating \textsc{HC-$\cF$} problems for classes of graphs $\cF$, which we call $\alpha\br{n}$-well-behaved. This means, that there exists a natural, cover-like ILP formulation of the \textsc{$p_{\cF}$-Partitioning} (\textsc{$p_{\cF}$-P}) problem (see Problem~\ref{prob:pfp}) as well as an algorithm which given a solution to the LP relaxation of this ILP, finds a solution to \textsc{$p_{\cF}$-P} with cost at most $\alpha\br{n}$ times the cost of the LP solution. This results in an $\bigo\br{\alpha\br{n}+\beta\br{n}}$-approximation algorithm for \textsc{HC-$\cF$}, where $\beta\br{n}$ is the approximation ratio of the LP-based bicriteria algorithm for \textsc{$\rho$-SEP} (Problem~\ref{prob:rho-sep}), currently at $\beta\br{n}=\bigo\br{\log n}$. As an immediate corollary we obtain an $\bigo\br{\log n\cdot\log\log n}$-approximation algorithm for clustering into trees and an $\bigo\br{\log n}$-approximation algorithm for clustering into clusters of diameter at most $d$. To the best of our knowledge, these are the first algorithmic results for the \textsc{HC-$\cF$} problem.

Conceptually, our work follows an implicitly, but broadly used technique which hereby we call \textit{flattening}. The idea is to subdivide the problem of finding a hierarchical/temporal\footnote{It is worth mentioning, that the hierarchical clustering tree can be treated as a sort of schedule of edges to be cut. This interpretation suggest that the depth of the tree can be treated as the time dimension of that schedule.} object such as a clustering tree into a problem of finding a sequence of `flat' objects such as separators or cuts. Among multiple uses of such a technique, one can mention \textsc{Hierarchical Clustering}~\cite{ACostFunctionForSimilarityBasedHierarchicalClustering, ApproximateHierarchicalClusteringViaSparsestCutAndSpreadingMetrics,Hierarchical_Clustering_Objective_Functions_and_Algorithms}, \textsc{Graph Search Problem}~\cite{Approximating_the_Average_Case_Graph_Search_Problem_with_Non_uniform_Costs}, \textsc{Tree Decompositions}~\cite{Approximating_Treewidth_Pathwidth_Frontsize_and_Shortest_Elimination_Tree}, \textsc{Precedence-Constrained Decision Tree}~\cite{Precedence-Constrained_Decision_Trees_and_Coverings} and many more.

Our ideas are particularly inspired by the work of Charikar and Chatziafratis~\cite{ApproximateHierarchicalClusteringViaSparsestCutAndSpreadingMetrics}, who used the flattening technique in conjunction with the \emph{spreading-metrics} framework. Whereas they used SDP relaxation, we encode the structure of the clustering tree using linear programming. The LP is divided into levels which allow decomposing it into several, level-restricted guides, which serve as a blueprint for the recursive construction of the clustering tree.
As it turns out, in order to encode the structural requirements given by the class $\cF$, we refine their technique by enriching the ILP formulation of the problem with generalized type of constraints which encode whether the clusters at each level have the desired properties. This is due to the fact that we need to force the ILP to optimize for two types of measures of progress: cutting edges required to transform clusters into $\cF$-clusters as well as efficient graph partitioning. Since $\cF$ might be an (almost) arbitrary class of graphs, we aim at showcasing a general framework of solving such problems, rather then focusing on singular solutions. To achieve this, we enforce certain natural conditions on the input family $\cF$, which in particular are satisfied by the graph classes of our interest. Any graph class respecting those conditions is called \emph{$\alpha\br{n}$-well-behaved}. Of particular importance is to us the notion of $\cF$-predicate which is a boolean function encoding whether a given vertex satisfies some structural property such that if all vertices of a graph satisfy the $\cF$-predicate, then all of the clusters are guaranteed to belong to $\cF$. Such a property can be for example whether a given vertex belongs to a cycle. If a given $\cF$ respects these conditions, we automatically get a series of constraints which will get included into our ILP, and this will suffice for a good approximation guarantee.

In order to extract useful information out of our ILP, we relax it into an LP and solve it (in polynomial time). Upon doing so, we use a simple arithmetic rounding procedure to fix the values of the $\cF$-predicate. Since our LP is organized into levels mimicking the levels of the hierarchy, we (recursively) decompose the solution according to those level and then use an LP-guided recursive processing in order to find the clustering tree.
At each recursion step, the algorithm uses the fragment of the LP on the current level to find two sets of edges to cut. The first, such that for each vertex $v$ in the current cluster, for which the LP value of the $\cF$-predicate is 1, the $\cF$-predicate is satisfied. The second, such that the remaining clusters are small enough. The algorithm then builds the rest of the clustering tree recursively. Using a careful lower bounding of the cost we get that the approximation ratio of our algorithm is $\bigo\br{\alpha\br{n}+\beta\br{n}}$.

To complement the algorithmic results, we show that for both $\cT$ and $\cD_d$ the problem cannot be approximated within any constant factor under the Small Set Expansion Hypothesis~\cite{GraphExpansionAndTheUniqueGamesConjecture}. We use a blackbox reduction from the inapproximability of \textsc{Hierarchical Clustering}~\cite{ApproximateHierarchicalClusteringViaSparsestCutAndSpreadingMetrics} by adding appropriate gadgets to all of the vertices of the original graph. Importantly, for the case of $\cF=\cT$ we use a structural property observed by the authors above, that the cost of the optimal solution to \textsc{HC} is lower bounded by the cost of the \textsc{Minimum Linear Arrangement}. Notably, this may no longer be true for the case for an arbitrary instance of \textsc{HC-$\cT$}, but we show that a similar lower bound still holds true for our construction.
\subsection{Related Work}

\paragraph{Flat Clustering}
Classical flat clustering objectives include \textsc{$k$-Center}, \textsc{$k$-Median}, and \textsc{$k$-Means}. For metric \textsc{$k$-Center}, a simple greedy farthest-first traversal yields a $2$-approximation, and this factor is optimal unless $P=NP$~\cite{GonzalezKCenter, HochbaumShmoysKCenter}. For metric \textsc{$k$-Median}, many constant factor approximation algorithms are known~\cite{AConstant-FactorApproximationAlgorithmForTheKMedianProblem, ApproximationAlgorithmsForMetricFacilityLocationAndKMedianProblemsUsingThePrimal-DualSchemaAndLagrangianRelaxation, ANewGreedyApproachForFacilityLocationProblems, AryaLocalSearchKMedian, Approximating_k-Median_via_Pseudo-Approximation, AnImprovedApproximationForKMedianAndPositiveCorrelationInBudgetedOptimization,BreachingThe2LMPApproximationBarrierForFacilityLocationWithApplicationsToKMedian} culminating in $\br{2+\epsilon}$ approximation of~\cite{A2-ApproximationAlgorithmForMetricKMedian}. For metric \textsc{$k$-Means}, constant-factor approximations are also known (for example via local search), with improved guarantees in Euclidean settings~\cite{KanungoKMeansLocalSearch, AhmadianKMeansImprovedApprox} with the currently best known guarantee of $3+2\sqrt{2}+\epsilon< 5.83$~\cite{AnImprovedGreedyApproximationForMetricKMeans}.

\paragraph{\textsc{Hierarchical Clustering}}
Work on \textsc{Hierarchical Clustering} can be split into similarity-based and dissimilarity-based objectives. In the similarity-based setting, Dasgupta introduced the now standard objective and initiated its approximation-theoretic study~\cite{ACostFunctionForSimilarityBasedHierarchicalClustering}. This line was further developed by Charikar and Chatziafratis, who obtained an $\bigo\br{\sqrt{\log n}}$-approximation via \textsc{Sparsest Cut} and spreading-metric SDP relaxations and showed that any constant-factor approximation is NP-hard under the Small Set Expansion hypothesis~\cite{ApproximateHierarchicalClusteringViaSparsestCutAndSpreadingMetrics}, and by the Auhhors in~\cite{Hierarchical_Clustering_Objective_Functions_and_Algorithms}, who studied objective-function variants and algorithmic guarantees. Another direction considers depth-based objectives, where one measures the depth of leaves or separators in the decomposition tree. For this type of objective, one gets a $\bigo\br{\sqrt{\log n}}$-approximation on trees~\cite{DereniowskiKosowskiUznanskiZouICALP2017}. For general graphs, combining a well known reduction from \textsc{Edge Ranking}~\cite{RankingsOfGraphs1998} to \textsc{Vertex Ranking} with the current state of the art approximation for the latter~\cite{Approximating_Treewidth_Pathwidth_Frontsize_and_Shortest_Elimination_Tree}, based on algorithm for\textsc{Vertex Separation} of~\cite{FeigeHajiaghayiLee2005} yields an $\bigo\br{\log^{3/2} n}$ guarantee.

In the dissimilarity-based setting, a dual objective to Dasgupta's cost is the \emph{revenue} objective. A sequence of approximation improvements for this objective includes $1/3$-approximation guarantees (for random splitting and average-linkage) by Moseley and Wang~\cite{MoseleyWangNIPS2017}, an SDP-based improvement to $0.3364$ by Charikar, Chatziafratis, and Niazadeh~\cite{CharikarChatziafratisNiazadehSODA2019}, a $0.4246$ guarantee via \textsc{Max-Uncut Bisection} in~\cite{BisectAndConquer2020}, and then a $0.585$-approximation by Alon, Azar, and Vainstein~\cite{AlonAzarVainsteinCOLT2020}.

\paragraph{\textsc{Feedback Edge/Vertex Set}}
\textsc{Feedback Edge Set} 
When $X=V$, the \textsc{FES} is equivalent to Maximum Spanning Tree which can be found by modifying the Kruskal's algorithm for the Minimum Spanning Tree. Otherwise, the problem is NP-hard and admits a $2$-approximation~\cite{ApproximatingMinimumSubsetFeedbackSetsInUndirectedGraphsWithApplications}.
For directed graphs, the problem is known as \textsc{Feedback Arc Set}, is NP-hard and admits an $\bigo\br{\log n \cdot \log\log n}$-approximation via an LP relaxation and rounding procedure in~\cite{ApproximatingMinimumFeedbackSetsAndMulticutsInDirectedGraphs}.
The Feedback Vertex Set problem is NP-hard even for $X=V$ for which it admits a $2$-approximation~\cite{A2-ApproximationAlgorithmForTheUndirectedFeedbackVertexSetProblem,OptimizationOfPearlsMethod} and a general $8$-approximation~\cite{An8ApproximationAlgorithmForTheSubsetFeedbackVertexSetProblem}. For directed graphs, the problem is equivalent to \textsc{FAS}.

\paragraph{\textsc{Balanced Partitioning}}
This area is tightly connected to \textsc{Sparsest Cut} and \textsc{Balanced Cut}: Leighton and Rao obtained an $\bigo\br{\log n}$-approximation for \textsc{Sparsest Cut} via multicommodity max-flow~\cite{LeightonRao1999}, later improved to $\bigo\br{\sqrt{\log n}}$ by Arora, Rao, and Vazirani via SDP and geometric embeddings~\cite{ExpanderFlowsGeometricEmbeddingsAndGraphPartitioning}. For \textsc{Minimum-Weight Balanced Vertex Separator}, Feige, Hajiaghayi, and Lee obtained an analogous $\bigo\br{\sqrt{\log n}}$-approximation~\cite{FeigeHajiaghayiLee2005}. Exact $(k,1)$-\textsc{Balanced Partitioning} (equal-size components) admits no polynomial-time approximation within any finite factor unless $P=NP$~\cite{AndreevRacke2004}, motivating the search of bicriteria approximation algorithms. Krauthgamer, Naor, and Schwartz gave an $\bigo\br{\sqrt{\log n}}$-approximation of the cut cost while allowing component sizes up to $2n/k$, based on an SDP relaxation~\cite{PartitioningGraphsIntoBalancedComponents}. For the \textsc{Minimum-Multicut} problem tightly connected to partitioning graph into balanced components, Garg, Vazirani, and Yannakakis obtained an $\bigo\br{\log n}$-approximation via a linear programming relaxation~\cite{ApproximateMaxFlowMinMultiCut1993}.

\section{Preliminaries}

Throughout this paper, each considered set of points along the similarity function is treated as an undirected weighted graph $G=\br{V,E,c,w}$, where $c\colon E\to \mathbb{N}$ is the similarity function (if $uv\notin E$, we may assume as a convention that $c(uv)=0$) and $w\colon V\to \mathbb{N}^+$ are vertex weights. We write $n\br{G}=\spr{V}$, $m\br{G}=\spr{E}$, and $n=n\br{G}$, $m=m\br{G}$ when $G$ is clear from context. For $U\subseteq V$, let $w\br{U}=\sum_{v\in U} w(v)$ and let $w\br{G}=w\br{V}$. For a set of edges $F\subseteq E$, we use $c\br{F}=\sum_{e\in F} c(e)$. Let $[n]=\brc{1,\dots, n}$.
By a permutation $\pi$ of $V$ we mean a bijection $\pi\colon V\to [n]$.

Consider a tree $T$ constructed for a graph $G$ in which each node $u$ corresponds to a \emph{cluster} $G_u$ (an induced subgraph of $G$), the root corresponds to $G$, and the children of a node $u$ form a partition of $G_u$ into connected components.
The set of edges of $G_u$ having endpoints in different connected components corresponding to the children of $u$ are denoted by $S_u$.
Depending on the problem we either require that the connected components corresponding to the leaves of $T$ are singletons in which case we refer to $T$ as a \emph{clustering tree}, or belong to the desired class $\cF$ and then $T$ is called a \emph{$\cF$-clustering tree}.

\subsection{Problem definitions}\label{subsec:problems-definitions}

For a given family of graphs $\cF$, let $p_{\cF}\colon V\to \brc{0,1}$ be a predicate function, which to each vertex assigns a Boolean value encoding whether a given vertex satisfies some structural property. We say that $p_{\cF}$ is an \emph{$\cF$-predicate} if for any graph $G$, we have that for every connected component $H$ of $G$, $H\in \cF$ whenever $\bigwedge_{v \in V}p_{\cF}\br{v}$ is true. Although this definition might seem a bit abstract, in what follows we show that the graph classes of our interest have naturally defined $\cF$-predicates. It should also be noted that for any graph class there might be many different $\cF$-predicates, but for our needs it is sufficient to consider only one of them per class.
Given a weighted graph $G=\br{V, E, c, w}$, we define the following combinatorial problems:
\begin{itemize}
    \item \textsc{Hierarchical Clustering} (\textsc{HC}): find a clustering tree $T$ of $G$ that minimizes $\COST_G\br{T} = \sum_{u\in V\br{T}} c(S_u) \cdot w\br{V\br{G_u}}$.
    \item \textsc{Hierarchical $\cF$-Clustering} (\textsc{HC-$\cF$}): find a $\cF$-clustering tree $T$ of $G$ that minimizes $\COST_G\br{T} = \sum_{u\in V\br{T}} c(S_u) \cdot w\br{V\br{G_u}}$.
    \item \textsc{$p_{\cF}$-Partitioning} (\textsc{$p_{\cF}$-P})\label{prob:pfp}:
    Given an additional set $X\subseteq V$ of terminals, find a subset of edges $E'\subseteq E$ such that for each $v\in X$, the $\cF$-predicate $p_{\cF}\br{v}$ is satisfied in $G\setminus E'$ and the cost $c\br{E'}$ is minimized.
    The following problems are special cases of \textsc{$p_{\cF}$-P}:
    \begin{enumerate}
        \item \textsc{Feedback Edge Set} (\textsc{FES}): given an additional set $X$ of vertices, find a subset of edges $E'\subseteq E$ such that $G\setminus E'$ contains no cycles going through vertices in $X$ and $c\br{E'}$ is minimized. Here, the $\cF$-predicate is simply whether there exists a cycle in $G\setminus E'$ going through $v$.
    \item \textsc{Min-Multicut} (\textsc{MMC}): given a set of vertex pairs $\{(s_i, t_i)\in V\times V\colon i\in [k]\}$, find a subset of edges $E'\subseteq E$ such that each pair $(s_i, t_i)$ is disconnected in $G\setminus E'$ and $c\br{E'}$ is minimized. Given an additional metric $\dist$ on $V$ and a threshold $d$, the problem of partitioning into clusters so that every terminal belongs to a cluster in which all vertices are within distance at most $d$ from this terminal (with no requirement on distances between non-terminal pairs of vertices) can be reduced to \textsc{MMC} by taking terminal pairs $(v,u)$ for all $v\in X$ and $u\in V$ with $\dist\br{v, u} > d$, because such pairs cannot lie in the same connected component of $G\setminus E'$. Thus, for a terminal $v\in X$, the $\cF$-predicate $p_{\cF}\br{v}$ states that if $H$ is the connected component of $G\setminus E'$ containing $v$, then $\dist\br{v, u} \le d$ for every $u\in V(H)$.
    \item \textsc{$\rho$-Separator} (\textsc{$\rho$-SEP})\label{prob:rho-sep}: given a parameter $\rho\in(0,1)$ and an additional set $X\subseteq V$ of vertices, find a subset of edges $E'\subseteq E$ such that every connected component $H$ of $G\setminus E'$ with $V(H)\cap X\neq\emptyset$ has the total vertex weight at most $\rho\cdot w\br{G}$ and $c\br{E'}$ is minimized. We call such an $E'$ a $(\rho,X)$-separator. Here, the $\cF$-predicate is whether the connected component $H\in G\setminus E'$ containing $v$ has the total vertex weight at most $\rho\cdot w\br{G}$. A bicriteria approximation algorithm for \textsc{$\rho$-SEP} with parameters $(\rho,\rho_0)$, where $0<\rho<\rho_0<1$, returns a subset of edges $E'$ such that $c\br{E'}$ is at most $\rho$ times more then the optimum and every connected component of $H\in G\setminus E'$ with $H\cap X\neq\emptyset$ has total vertex weight of at most $\rho_0\cdot w\br{G}$.
    \end{enumerate}
    \item \textsc{Minimum Linear Arrangement} (\textsc{MLA}): find a permutation $\pi$ of the vertices of $G$ that minimizes $\COST_G\br{\pi} = \sum_{uv\in E} c(uv) \cdot \spr{\pi(u) - \pi(v)}$.
\end{itemize}

The inclusion of weights serves the purpose of obtaining a more general solution, however its occurrence does not impact the techniques nor the results in a significant way.
On the other hand, we point out that our working assumption will be that the value $w\br{G}$ is polynomial in $n$.
If this assumption is not met for a given input, then the algorithms work correctly but become pseudo-polynomial in the size of the graph.

For any node $u\in V(T)$ of a clustering (or $\cF$-clustering) tree $T$, we denote by $\COST_G\br{T, u} = c(S_u) \cdot w\br{V\br{G_u}}$ the contribution of the node $u$ to the total cost of the tree.
By $\OPT_{\cF}$ we denote the cost of an optimal solution to \textsc{HC-$\cF$}.

We further remark, that the introduction of $\cF$-predicates allows us to use somehow weaker conditions to encode progress of our algorithm while trying to find a solution to \textsc{HC-$\cF$}. In particular, consider the case $\cF=\cT$. Here, there are at least two natural choices for $p_{\cF}$. For our needs, we want that there is no cycle in $G$ going through a given vertex $v\in V$. Observe however, that we also could define the $p_{\cF}$ to encode whether the cluster containing $v$ is a tree. The latter condition is stronger, but the former one allows us to use already existing approximation algorithms for \textsc{FES}. As it turns out, either way our ILP formulation will be forced to cut edges of a given cluster if at least one vertex of this cluster does not satisfy the $\cF$-predicate. Therefore, we can use the weaker condition to encode progress of our procedure which greatly eases the design of the algorithm. This relaxed condition also potentially allows for finding more classes of graphs which fit our framework.

\section{Well-behaved graph classes}

In this section we provide a series of characteristics of the class $\cF$ which will enable constructing a general algorithmic framework for approximating \textsc{HC-$\cF$}. Although this characterization may seem quite technical, it is satisfied by natural graph classes including trees and bounded diameter graphs (See Section \ref{sec:applying-framework}). The idea is that we want a \textsc{$p_{\cF}$-P} problem to admit a certain kind of natural ILP formulation with an LP relaxation solvable in polynomial time as well as a rounding algorithm for this LP with some provable approximation guarantee. If this is the case, then it will be easy for us to encode all of the necessary constraints into our general linear programming framework. The following definition captures the above intuitions in a formal way.
Any class of graphs $\cF$ that satisfies the following properties will be called \emph{$\alpha\br{n}$-well-behaved}:
\begin{enumerate}
    \item There exists an ILP formulation of the \textsc{$p_{\cF}$-P} problem such that: the set of variables consists only of variables of edge-cut type. An \emph{edge-cut type} variable $y_e$ for an edge $e\in E$ is an indicator of the fact that the ILP chooses $e$ to be cut. 
    \item The requirement to satisfy the $\cF$-predicate can be encoded using cover-type constraints and these are the only existing constraints (except the domain of each variable). A \emph{cover-type} constraint is a constraint of the form $\sum_{e\in R} y_e \geq 1$ for some subset of edges $R\subseteq E$. For any $v\in X$, let $\cR_v$ denote the set of all such $R$, for which a cover-type constraint exists and is necessary in order for $p_{\cF}\br{v}$ to be true. To be exact, $p_{\cF}\br{v}$ iff $\forall_{R\in \cR_v} \sum_{e\in R} y_e \geq 1$.
    \item The objective function is $\sum_{e\in E} c\br{e}\cdot y_e$.
    \item There exists a separation oracle for the LP-relaxation of the above ILP.
    \item There exists an algorithm, which given an instance of the problem \textsc{$p_{\cF}$-P} and a solution to the LP relaxation of the above ILP, finds a subset of edges $E'$ such that $G\setminus E'$ is a valid solution to \textsc{$p_{\cF}$-P} and $\sum_{e\in E'} c(e) \leq \alpha(n)\cdot \sum_{e\in E} c(e) \cdot y_e$, where $\alpha\br{n}$ is the approximation ratio of the algorithm.
\end{enumerate}

From now on by \textsc{$p_{\cF}$-P-ILP} and \textsc{$p_{\cF}$-P-LP} we denote any such ILP and its LP-relaxation for \textsc{$p_{\cF}$-P}. It should be noted that whenever there are multiple choices for the ILP and its LP-relaxation, one can pick arbitrarily among them, and by a slight abuse of notation the above symbols represent any of them/the one desirable for our purposes.

\section{General algorithmic framework for approximating \textsc{HC-$\cF$}}

Having characterized the kinds of graph classes that will fit into our framework we can now move on towards describing how to use the properties of $\alpha\br{n}$-well-behavedness to construct a general approximation algorithm for \textsc{HC-$\cF$}. The main idea is that we use a level-based \textsc{LP} relaxation to guide the recursive algorithm.
In the preprocessing step we encode the problem using an \textsc{ILP}. We obtain this formulation by combining any $\alpha\br{n}$-well-behaved \textsc{ILP} with the level decomposition approach from \cite{ApproximateHierarchicalClusteringViaSparsestCutAndSpreadingMetrics}. Our \textsc{ILP} consists of multiple variable levels which encode levels of an $\cF$-clustering tree. It should be noted that these levels are indexed by the size of the clusters, rather than the depth of the tree.

After solving \textsc{LP-HC-$\cF$}, we use a rounding routine to get a partially integral solution which at each level corresponds to choosing which vertices should satisfy the $\cF$-predicate.

During the recursive step, we process clusters in a top-down manner starting from $H=G$. For a cluster $H$, we set $t=w\br{V\br{H}}/2$. Then, we look which vertices satisfy the $\cF$-predicate at level $t$ and which do not (For a formal level definition see Section~\ref{sec:levels} below).
According to this, we partition them into two sets $X_1$ (satisfying the predicate) and $X_0$ (not satisfying the predicate). The solution of \textsc{LP-HC-$\cF$} at level $t$ is then used as a valid solution to \textsc{$p_{\cF}$-P-LP} with $X=X_1$, yielding a cut $E_{\cF}$ with $\alpha\br{n}$ factor loss. The same level-$t$ fragment is also used as a valid guide for \textsc{$\rho$-SEP} with $\rho=1/2$, $\rho_0=2/3$ and $X=X_0$, yielding a cut $E_{\rho}$, such that each connected component of $H\setminus E_{\rho}$ that contains a terminal, has total weight at most $2w\br{V\br{H}}/3$ and the cut cost is within a factor $\beta\br{n}=\bigo\br{\log n}$ of the guide. After cutting the edges in $E_{\cF}\cup E_{\rho}$, we recursively process all connected components.

Since we decompose the \textsc{LP} into separate levels, by summing up over all of them we will be able to show that our algorithm achieves an approximation ratio of $\bigo\br{\alpha\br{n}+\beta\br{n}}$.

\subsection{Level decomposition of \textsc{HC-$\cF$}} \label{sec:levels}

In order to analyze the structure of $\cF$-clustering trees, we introduce the notion of \emph{level decomposition}, which for a given $\cF$-clustering tree $T$ is defined as follows.
Let $\cC^T$ denote the family of all clusters in $T$, $\cC^T=\brc{G_u\colon u\in V(T)}$.
For $t\in [w\br{G}]$, denote by $\cL_t^T$ the set of all maximal clusters in $\cC^T$ of total weight at most $t$. We call $\cL_t^T$ the $t$-th \emph{level} of the decomposition.
Observe that for each $t$, $\cL_t^T$ is a partition of $V\br{G}$. By $E_t^T$ denote the set of edges with endpoints in different clusters of $\cL_t^T$. We call $E_t^T$ the $t$-th \emph{level cut}.
It is easy to see that the cost of $T$ can be expressed as:
\begin{observation}\label{obs:cost-decomposition}
    $\COST_G\br{T} = \sum_{t=1}^{w\br{G}} c\br{E_t^T}$.
\end{observation}

Since we use the above levels to decompose a $\cF$-clustering tree and its cost, we would also like to do the same for the $\cF$-predicate. This is because eventually, all vertices satisfy the $\cF$-predicate, but whether a vertex $v\in V$ satisfies the $\cF$-predicate, when it belongs to a cluster at the $t$-th level depends on the structure of the $\cF$-clustering tree.
For every level $t\in [w\br{G}]$ and every vertex $v\in V\br{G}$, we define the level-indexed $\cF$-predicate indicator by setting $p_{\cF}^t\br{v}$ to be true iff $p_{\cF}\br{v}$ is true in $G\setminus E_t^T$. 
\subsection{\textsc{ILP} formulation of \textsc{HC-$\cF$} and its relaxation}
    We now move towards a description of the constraints of the \textsc{ILP} along with an explanation of the role they play. We use the well-known \emph{spreading-metrics} framework which in recent years have come particularly useful for approximating many problems regarding graph cutting \cite{ExpanderFlowsGeometricEmbeddingsAndGraphPartitioning,PartitioningGraphsIntoBalancedComponents,ApproximateHierarchicalClusteringViaSparsestCutAndSpreadingMetrics,FastApproximateGraphPartitioningAlgorithms,DivideAndConquerApproximationAlgorithmsViaSpreadingMetrics}.
    Spreading metrics are a way to encode the structure of a graph into a metric space, which can then be used to find cuts in the graph. The task of a mathematical program (usually LP or SDP) is to find a metric that spreads out the vertices of the graph in a way that reflects the structure of the graph. Here, the idea is to merge such constraints with our level-indexed formulation to encode \textsc{HC-$\cF$} as a linear program.

    For each edge $e\in E$ and level $t\in [w\br{G}]$ we use a variable $y_e^t\in[0,1]$ which encodes whether $e\in E_t^T$.
    We also define a variable $x_v^t=1$ iff $p_{\cF}^t\br{v}$ is true and $0$ otherwise.  Let $\dist_{y, V}^{t}\colon V\times V\to \mathbb{R}_{\geq 0}$ denote shortest-path distances in $G$ induced by lengths $\{y_e^t\}_{e\in E}$.
    We first explain all constraints and then we give the full ILP. 
    These are as follows:
    \begin{enumerate}
        \item \textbf{refinement:} $y_{e}^t \leq y_{e}^{t-1}$ for all $e\in E$ and $t\in [w\br{G}]$, because the partition at level $t-1$ refines the partition at level $t$. Moreover, $x_v^t \leq x_v^{t-1}$ for all $v\in V$ and $t\in [w\br{G}]$, since $p_{\cF}^t\br{v}$ implies $p_{\cF}^{t-1}\br{v}$.

        \item \textbf{$p_{\cF}$-P constraints:} Recall, that for every vertex $v\in V$, the set $\cR_v$ contains all subsets of edges $R\subseteq E$ such that if at least one edge in each $R\in \cR_v$ is cut, then $p_{\cF}\br{v}$ is true. Therefore, for each $v\in V$, every level $t\in [w\br{G}]$ and every $R\in \cR_v$ we have $\sum_{e\in R} y_e^t \geq x_{v}^t$. Additionally, we have $x_v^1=1$, since $p_{\cF}^1\br{v}$ is true for all $v\in V$.

        \item \textbf{spreading-metrics:} For every vertex $v\in V$, every level $t\in [w\br{G}]$, and every subset $S\subseteq V$ containing $v$, we have
        \[
            \sum_{u\in S} w(u)\cdot \dist_{y, V}^{t}(v,u) \geq \br{1-x_v^t}\cdot \br{w\br{S}-t}.
        \]
        When $x_v^t=0$, this condition enforces that inside any set $S$ containing $v$, enough weighted mass is separated from $v$ unless the cluster weight around $v$ is already at most $t$. When $x_v^t=1$, the right-hand side vanishes, corresponding to the fact that the cluster containing $v$ is an $\cF$-cluster and hence its weight can be larger than $t$.
    \end{enumerate}

    The following \textsc{ILP}, which we call \textsc{ILP-HC-$\cF$}, is an exact formulation of the \textsc{HC-$\cF$} problem:
    \begin{align}
        \min \quad &\sum_{t=1}^{w\br{G}} \sum_{e\in E} c(e) \cdot y_{e}^t\\
        \text{subject to} \quad &y_{e}^t \leq y_{e}^{t-1} \quad \forall e\in E,\; t\in [w\br{G}],\\
        &x_v^t \leq x_v^{t-1} \quad \forall v\in V, t\in [w\br{G}],\\
        &x_v^1 = 1 \quad \forall v\in V,\\
        &\sum_{e\in R} y_{e}^t \geq x_{v}^t \quad \forall v\in V, R\in \cR_v, t\in [w\br{G}], \label{cover_constraint_ILP}\\
        &\sum_{u\in S} w(u)\cdot \dist_{y, V}^{t}(v,u) \geq \br{1-x_v^t}\cdot \br{w\br{S}-t} \quad \forall S\subseteq V, v\in S, t\in [w\br{G}], \label{spreading_constraint_ILP}\\
        &y_{e}^t \in \{0,1\} \quad \forall e\in E, t\in [w\br{G}],\\
        &x_v^t \in \{0,1\} \quad \forall v\in V, t\in [w\br{G}].
    \end{align}
    We denote by $\OPT_{ILP}$ the value of an optimal solution to the above ILP for a given graph $G$.

    \begin{lemma}\label{lem:ilp-exact}
        \textsc{ILP-HC-$\cF$} is an exact formulation of \textsc{HC-$\cF$} which implies that $\OPT_{ILP}=\OPT_{\cF}$.
    \end{lemma}
    \begin{proof}
        $(\Rightarrow)$ Let $T$ be any feasible $\cF$-clustering tree.
        Recall that, for each level $t\in [w\br{G}]$, $E_t^T$ is the level cut from the level decomposition. Define a solution by setting
        \[
            y_e^t:=\mathbf{1}[e\in E_t^T],
        \]
        and $x_v^t:=1$ if and only if the level-$t$ cluster of $v$ satisfies $p_{\cF}$.

        The refinement constraints hold because level cuts are nested. The constraints on $x$ hold by monotonicity of the predicate along the refinement. The $p_{\cF}$-P constraints hold by definition of the family $\cR_v$: whenever $x_v^t=1$, each required $p_{\cF}$-P constraint for $v$ at level $t$ is satisfied by the cut pattern encoded by $y^t$.

        For spreading constraints, fix $v\in V$, $t\in [w\br{G}]$, and $S\subseteq V$ with $v\in S$. If $x_v^t=1$, the right-hand side in~\eqref{spreading_constraint_ILP} is $0$. Hence assume $x_v^t=0$ and let $C$ be the level-$t$ cluster containing $v$. Then $C$ is not an $\cF$-cluster, so by level definition $w(C)\le t$. Every $u\in S\setminus C$ is separated from $v$ by level-$t$ cuts, hence $\dist_{y, V}^{t}(v,u)\ge 1$. Therefore,
        \[
            \sum_{u\in S} w(u)\cdot \dist_{y,V}^{t}(v,u)
            \ge \sum_{u\in S\setminus C} w(u)
            = w(S)-w(S\cap C)
            \ge w(S)-w(C)
            \ge w(S)-t,
        \]
        which is exactly \eqref{spreading_constraint_ILP} for $x_v^t=0$. Thus $(\mathbf{x},\mathbf{y})$ is ILP-feasible and
        \[
            \sum_{t=1}^{w\br{G}}\sum_{e\in E} c(e)\cdot y_e^t
            = \sum_{t=1}^{w\br{G}} c(E_t^T)
            = \COST_G(T)
        \]
        by Observation~\ref{obs:cost-decomposition}. Hence $\OPT_{ILP}\le \OPT_{\cF}$.

        $(\Leftarrow)$ Let $(\mathbf{x},\mathbf{y})$ be any integral feasible ILP solution. For each level $t$, define
        \[
            E_t^0:=\brc{e\in E\colon y_e^t=0},
        \]
        and let $\Pi_t$ be the partition of $V$ into connected components of $G\setminus \br{E\setminus E_t^0}$. (Note that $E\setminus E_t^0$ are exactly the edges that are considered to be cut at level $t$). Because $y_e^t\le y_e^{t-1}$, we have $E_{t-1}^0\subseteq E_t^0$, so $\Pi_{t-1}$ refines $\Pi_t$ and the family $(\Pi_t)_{t=1}^{w\br{G}}$ is laminar; therefore it defines a $\cF$-clustering tree $T$.

        Consider any level $t$ and any component $C\in \Pi_t$. For every $v\in C$, all vertices of $C$ are at zero $\dist_{y,V}^t$-distance from $v$, so plugging $S=C$ into \eqref{spreading_constraint_ILP} gives
        \[
            0\ge (1-x_v^t)\cdot (w(C)-t).
        \]
        Hence either $w(C)\le t$, or $x_v^t=1$ for all $v\in C$. In the second case, by $p_{\cF}$-P constraints and the definition of $\cR_v$, each vertex in $C$ satisfies the $\cF$-predicate; therefore $C\in \cF$. So every cluster in level $t$ is either of weight at most $t$ or belongs to $\cF$, exactly as required in a $\cF$-clustering tree.

        Finally, the level cut induced by $\Pi_t$ is precisely $\brc{e\in E\colon y_e^t=1}$, so
        \[
            \COST_G(T)=\sum_{t=1}^{w\br{G}} c\br{\brc{e\in E\colon y_e^t=1}}
            =\sum_{t=1}^{w\br{G}}\sum_{e\in E} c(e)\cdot y_e^t.
        \]
        Therefore every integral ILP solution maps to a feasible $\cF$-clustering tree of the same cost, implying $\OPT_{\cF}\le \OPT_{ILP}$, which completes the proof.
    \end{proof}

    We now wish to show that it is possible to solve a relaxation of \textsc{ILP-HC-$\cF$} in polynomial time. We denote this relaxation by \textsc{LP-HC-$\cF$}, which is obtained by replacing the integrality constraints on $x_v^t$ and $y_e^t$ with $x_v^t\in [0,1]$ and $y_e^t\geq 0$ for all $v\in V$, $e\in E$, and $t\in [w\br{G}]$.
    \begin{lemma}\label{lem:lp-polytime}
        The relaxation \textsc{LP-HC-$\cF$} can be solved in polynomial time.
    \end{lemma}
    \begin{proof}
        First, we recall our assumption on $w\br{G}$ being polynomial in $n$.
        Thus, the only potentially exponential families of constraints are \eqref{cover_constraint_ILP} and \eqref{spreading_constraint_ILP}. For constraints of type \eqref{cover_constraint_ILP}, a polynomial-time separation oracle is part of the assumptions of a well-behaved formulation.

        It remains to separate the spreading-metrics constraints. Fix a level $t\in [w\br{G}]$ and a vertex $v\in V$. Using the current values of the variables $y_e^t$, we compute all distances $\dist_{y,V}^{t}(v,u)$ for $u\in V$ by a shortest-path algorithm. Then the left-hand side minus the right-hand side of \eqref{spreading_constraint_ILP} can be rewritten as
        \[
            \sum_{u\in S} w(u)\cdot \dist_{y,V}^{t}(v,u) - \br{1-x_v^t}\cdot \br{w\br{S}-t}
            = \br{1-x_v^t}\cdot t - \sum_{u\in S} w(u)\cdot \br{\br{1-x_v^t}-\dist_{y,V}^{t}(v,u)}\geq 0.
        \]
        Therefore, for fixed $v$ and $t$, the most violated constraint is obtained by taking
        \[
            S_{v,t}=\brc{u\in V\colon \dist_{y,V}^{t}(v,u)<1-x_v^t},
        \]
        since every vertex with positive contribution to the second term should be included, while vertices with non-positive contribution should be excluded. Hence, after one shortest-path computation, we can test in polynomial time whether any spreading-metrics constraint for $S=S_{v,t}$ is violated, and if so, we return it as a separating hyperplane.

        Repeating this for all polynomially many pairs $(v,t)$ yields a polynomial-time separation oracle for constraints of type \eqref{spreading_constraint_ILP}. Consequently, \textsc{LP-HC-$\cF$} admits a polynomial-time separation oracle, and thus it can be solved in polynomial time by the ellipsoid method.
    \end{proof}
    Let $LP$ be a solution to \textsc{LP-HC-$\cF$} and let $\OPT_{LP}$ be the optimal value of its objective function. We now show how to round the solution of \textsc{LP-HC-$\cF$} to get a partially integral solution without excessive increase in cost.

    \begin{lemma}\label{lem:lp-rounding}
        Round the values of $x_v^t$ to get $\hat{x}_v^t = 1$ if $x_v^t \geq 1/2$ and $\hat{x}_v^t = 0$ otherwise
        and $\hat{y}_{e}^t=2\cdot y_{e}^t$. Then, $\br{\hat{\mathbf{x}}, \hat{\mathbf{y}}}$ is a feasible solution for \textsc{LP-HC-$\cF$} with cost at most $2\cdot \OPT_{LP}$.
    \end{lemma}
    \begin{proof}
        Firstly, observe that the refinement constraints are satisfied since $\hat{y}_{e}^t \leq \hat{y}_{e}^{t-1}$ and $\hat{x}_v^t \leq \hat{x}_v^{t-1}$ for all $e\in E$, $v\in V$ and $t\in [w\br{G}]$. The $p_{\cF}$-P constraints are satisfied since if $\hat{x}_v^t=1$, then $x_v^t\geq 1/2$ and hence $\sum_{e\in R} y_{e}^t \geq x_v^t \geq 1/2$ for all $R\in \cR_v$. Therefore, $\sum_{e\in R} \hat{y}_{e}^t = 2\cdot \sum_{e\in R} y_{e}^t \geq 1$. If $\hat{x}_v^t=0$, then the $p_{\cF}$-P constraint is trivially satisfied. The spreading-metrics constraints are satisfied since if $\hat{x}_v^t=0$, then $x_v^t<1/2$ and hence
        \[
            \sum_{u\in S} w(u)\cdot \dist_{y,V}^{t}(v,u) \geq \br{1-x_v^t}\cdot \br{w\br{S}-t} > 1/2\cdot \br{w\br{S}-t}.
        \]
        Therefore, we get that for all $S\subseteq V$ containing $v$, such that $\hat{x}_v^t=0$, we have
        \[
            \sum_{u\in S} w(u)\cdot \dist_{\hat{y},V}^{t}(v,u) = 2\cdot \sum_{u\in S} w(u)\cdot \dist_{y,V}^{t}(v,u) > w\br{S}-t
        \]
        as required.
        If $\hat{x}_v^t=1$, then the spreading-metrics constraint is trivially satisfied.
        Finally, the cost of the new solution is at most $2\cdot \OPT_{LP}$ since $\sum_{t=1}^{w\br{G}} \sum_{e\in E} c(e) \cdot \hat{y}_{e}^t = 2\cdot \sum_{t=1}^{w\br{G}} \sum_{e\in E} c(e) \cdot y_{e}^t = 2\cdot \OPT_{LP}$.
    \end{proof}

    Overall, we obtain a partially integral point where all of the $x$ values are integral, which translates to the fact that we have decided which vertices need to satisfy the $\cF$-predicate at each level of the $\cF$-clustering tree. From now on we will thus assume that for every $v\in V$ and $t\in[w\br{G}]$, the value of $p_{\cF}^t\br{v}$ is fixed and we can remove all of the $x$ variables from the \textsc{LP}. Since we will not use the original values in further considerations, by a slight abuse of notation from now on we will refer to this solution simply by $LP$.

\subsection{LP relaxation for $\rho$-separators}\label{sec:lp-rho-sep}

    We now state the linear programming relaxation for the \textsc{$\rho$-SEP} which is a modified LP of \cite{FastApproximateGraphPartitioningAlgorithms} that will be required in order to reduce the size of the instance efficiently. For the \textsc{$\rho$-SEP} problem on a graph $G=(V,E)$ with a terminal set $X\subseteq V$, we use the following weighted spreading-metrics \textsc{LP} relaxation \textsc{LP-$\rho$-SEP}:
    \begin{align}
        \min \quad &\sum_{e\in E} c(e)\cdot y_e\\
        \text{subject to} \quad &\sum_{u\in S} \dist_{y,V}(v,u)\cdot w(u) \geq w\br{S}-\rho\cdot w\br{G} \quad \forall S\subseteq V,\; v\in S\cap X,\\
        &0\leq y_e\quad \forall e\in E.
    \end{align}

    The analysis in \cite{FastApproximateGraphPartitioningAlgorithms} is stated for the special case $X=V$. In our framework, we use the corresponding extension for the case $X\neq V$: the same sphere-growing argument is applied, the only important modification is that we grow the spheres around the centers in $X$, rather than around arbitrary vertices in $V$. We state the resulting guarantee below and defer the full extension from $X=V$ to $X\neq V$ to Appendix~\ref{app:terminal-extension}.

    The following bicriteria rounding guarantee for \textsc{$\rho$-SEP} is a direct corollary of \cite{FastApproximateGraphPartitioningAlgorithms}:
    \begin{theorem}\label{thm:ksep-rounding}
        Let $\rho$ and $\rho_0$, $0<\rho<\rho_0<1$, be any fixed constants.
        There exists a polynomial-time algorithm which, given an instance of \textsc{$\rho$-SEP} and a solution to the above \textsc{LP-$\rho$-SEP}, returns a subset of edges $E'$ such that every connected component of $G\setminus E'$ that contains a terminal has total weight at most $\rho_0\cdot w\br{G}$ and
        \[
            \sum_{ij\in E'} c\br{i,j} \leq \bigo\br{\log n}\cdot \sum_{e\in E} c(e)\cdot y_e.
        \]
    \end{theorem}

\subsection{Restricting the solutions}
    
In order to obtain quality approximation ratios we employ a charging scheme, where we split the $LP$ into several solutions restricted to the clusters of the $\cF$-clustering tree $T$ we are building. Fix some cluster $H$ of $G$ and a level $t$. By $LP_t^H$ denote the solution of \textsc{LP-HC-$\cF$} restricted to the set of vertices $H$ at level $t$ which is done by discarding all variables $y_{uv}^z$, such that $\brc{u,v}\cap V\br{H}\neq \brc{u,v}$ or $z\neq t$. We call each such restricted part $LP_t^H$ a \emph{guide}. For any target problem $\Pi$ with LP relaxation \textsc{LP-$\Pi$} (on the same graph fragment and terminal set), we say that $LP_t^H$ is a \emph{valid guide} for $\Pi$ if the restricted values from $LP_t^H$ satisfy all constraints of \textsc{LP-$\Pi$}. In other words, a valid guide is simply a feasible LP point for the target problem, inherited from \textsc{LP-HC-$\cF$}. This restriction is also performed on the cost function and this is reflected in the accordingly defined quantities $\COST\br{LP_t^H} = \sum_{uv\in E\br{H}} c\br{u,v} \cdot y_{uv}^t$. We have the following lemma:

    \begin{lemma}\label{lem:relaxation-decomposition}
        Let $T$ be any $\cF$-clustering tree. For any cluster $H \in \cC^T$, let $r\br{H}$ be the weight of the heaviest child cluster of $H$ if it exists or $0$ otherwise. Moreover, fix some constant $g\in \mathbb{N}$. Then:
        \begin{align*}
            \sum_{H\in \cC^T}\sum_{t=r\br{H}+1}^{w\br{V\br{H}}} c\br{LP_{\fl{t/g}}^H}\leq g\cdot \OPT_{LP}
        \end{align*}
    \end{lemma}
    \begin{proof}
        For convenience, define $y_{u}^0=y_{u}^1$ for all $u\in V$. The value $c\br{LP_{0}}$ is defined analogously as before.
        We start by showing that 
        \[
        \sum_{H\in \cC^T}\sum_{t=r\br{H}+1}^{w\br{V\br{H}}} c\br{LP_{\fl{t/g}}^H}\leq \sum_{i=0}^{w\br{G}} c\br{LP_{\fl{i/g}}^G}.
        \]
        It suffices to show that for any fixed edge $uv\in E$ and level $t\in [w\br{G}]$, the term $c\br{u,v}\cdot y_{uv}^t$ is counted at most once in the above sum. Assume for contradiction that there are two clusters $H$ and $B$ both contributing this term at level $t$. By definition of hierarchical clustering, one of $H$, $B$ must be an ancestor of the other; without loss of generality assume $H$ is an ancestor of $B$. Then
        \[
            t \in \left(r\br{H}, w\br{V\br{H}}\right]\quad \text{and} \quad
            t \in \left(r\br{B}, w\br{V\br{B}}\right].
        \]
        Since $H$ is an ancestor of $B$, we have $r\br{H} \geq w\br{V\br{B}}.$
        Hence the above intervals are disjoint, a contradiction. The bound follows by summing over all edges and levels.

        Now, observe that by simple reindexing, we have that 
        \[\sum_{i=0}^{w\br{G}} c\br{LP_{\fl{i/g}}^G}\leq g\cdot\sum_{i=1}^{w\br{G}} c\br{LP_{i}^G}.\]
        But the right-hand side is exactly $g\cdot \OPT_{LP}$, which completes the proof.
    \end{proof}

    \begin{lemma}\label{lem:lp-level-bound}
        Let $H$ be any cluster and let $d, b\in\mathbb{N}^+$. Then
        \[
        w\br{V\br{H}}\cdot c\br{LP_{\fl{w\br{V\br{H}}}/b}^H}\leq d\cdot \sum_{t=\br{1-1/d}\cdot w\br{V\br{H}}+1}^{w\br{V\br{H}}} c\br{LP_{\fl{t/b}}^H}.
        \]
    \end{lemma}
    \begin{proof}
    The inequality follows from monotonicity of level cuts: as $t$ decreases, more edges belong to the cut.
    \end{proof}

    Fix a subgraph $H$ of $G$ with weight $r=w\br{V\br{H}}$ and a level $t\in [r]$. Define
    \[
        X_1:=\brc{v\in V\br{H}\colon \hat{x}_v^t=1}, \qquad X_0:=\brc{v\in V\br{H}\colon \hat{x}_v^t=0}.
    \]

    The next two lemmas allows us to decompose a solution to \textsc{LP-HC-$\cF$} into several instances of \textsc{LP-$p_{\cF}$-P} and \textsc{LP-$\rho$-SEP} with terminal sets $X_1$ and $X_0$, respectively.
    \begin{lemma}\label{lem:lp-fes}
        Let $H$ be cluster that does not belong to $\cF$, with weight $r=w\br{V\br{H}}$ and let $t\leq r$ be some fixed level. Then $LP_r^H$ is a valid solution for the \textsc{$p_{\cF}$-P} for $X=X_1$.
    \end{lemma}
    \begin{proof}
        In $LP_r^H$ there are additional metric/spreading constraints for vertices $v$ such that $p_{\cF}^t(v)=0$. This is acceptable since additional constraints yield a stricter relaxation. Now consider any $v\in V$. If $p_{\cF}^t(v)=0$, then the $p_{\cF}$-P constraint is always satisfied, which corresponds to the fact that $v$ does not need to belong to a $\cF$-cluster. If $p_{\cF}^t(v)=1$, then the $p_{\cF}$-P constraint is satisfied since it is explicitly included in the \textsc{LP-HC-$\cF$} formulation. Moreover, since $v\in X_1$, we have that all $p_{\cF}$-P constraints in the \textsc{LP-$p_{\cF}$-P} formulation are included in the \textsc{LP-HC-$\cF$} formulation. 
        We note, that whenever there exists a variable $y_e^t>1$, we can always set it $1$ while plugging it into the \textsc{LP-$p_{\cF}$-P} formulation (without increase in the cost), since the constraints will still be satisfied\footnote{Here, the technical issue is that during our partial rounding step some values of $y_e^t$ might exceed $1$.}.
        Therefore, the claim follows.
    \end{proof}
    \begin{lemma}\label{lem:lp-ksep}
        Let $H$ be a non $\cF$-cluster of weight $r=w\br{V\br{H}}$ and let $t<r$. Then $LP_t^H$ is a valid guide for the weighted spreading-metrics \textsc{LP-$\rho$-SEP} with $\rho=t/r$ for $X=X_0$.
    \end{lemma}
    \begin{proof}
        The statement is analogous to Lemma~\ref{lem:lp-fes}, but with the complement terminal set.
        First, the $p_{\cF}$-P constraints are irrelevant for \textsc{$\rho$-SEP} feasibility and can only make the relaxation stronger.

        Second, consider the vertices in $X_0$ (i.e., those with $p_{\cF}^t(v)=0$). For every such $v$ and every $S\subseteq V(H)$ with $v\in S$, the spreading-metrics constraint of \textsc{LP-HC-$\cF$} at level $t$ gives
        \[
            \sum_{u\in S} w(u)\cdot \dist_{y,V}^{t}(v,u) \geq w\br{S}-t,
        \]
        where $\dist_{y,V}^{t}$ denotes shortest-path distances taken in the whole graph $G$.

        On the other hand, in \textsc{LP-$\rho$-SEP} for the cluster $H$, the relevant distances are the ones measured inside $H$, namely $\dist_{y,V(H)}^{t}$. Since $H$ is a subgraph of $G$, every $u$-$v$ path contained in $H$ is also a path in $G$, whereas paths leaving $H$ are no longer available; restricting the graph to the cluster $H$ can therefore only remove paths and never shorten the distances. Consequently,
        \[
            \dist_{y,V(H)}^{t}(v,u) \geq \dist_{y,V}^{t}(v,u) \qquad \text{for all } u,v\in V(H).
        \]
        Combining the two inequalities, for every $S\subseteq V(H)$ and every $v\in S\cap X_0$ we obtain
        \[
            \sum_{u\in S} w(u)\cdot \dist_{y,V(H)}^{t}(v,u) \geq \sum_{u\in S} w(u)\cdot \dist_{y,V}^{t}(v,u) \geq w\br{S}-t.
        \]
        With $t$ and $V\br{H}$ as the metric, this is exactly the weighted spreading requirement of \textsc{LP-$\rho$-SEP} corresponding to $\rho=t/r$. Hence, all constraints of \textsc{LP-$\rho$-SEP} are satisfied, and $LP_t^H$ is a feasible \textsc{LP-$\rho$-SEP} guide for $X_0$.
    \end{proof}

    \section{The Algorithm}

    In this section, we use the \textsc{$\rho$-SEP} parameters $\rho=1/2$ and $\rho_0=2/3$, and we assume that there exists an algorithm that for any solution for the $\textsc{LP-$\rho$-SEP}$, returns a set of edges of cost at most $\beta\br{n}$ times the \textsc{LP} cost (See Theorem \ref{thm:ksep-rounding} for a procedure with a guarantee $\beta\br{n}=\bigo\br{\log n}$). The main algorithm which includes the preprocessing stage is given in Algorithm \ref{alg:main-procedure} and the recursive sub-procedure, which performs the rounding step is given in Algorithm \ref{alg:recursive-procedure}. The main procedure computes the partially rounded \textsc{LP} solution $(\hat{\mathbf{x}},\hat{\mathbf{y}})$ and then invokes the recursive routine on the entire graph $G$. The recursive routine only builds the $\cF$-clustering tree, using the guides prepared by the main procedure. In particular, it uses the partially rounded \textsc{LP} solution $(\hat{\mathbf{x}},\hat{\mathbf{y}})$ to find a set of edges $E_{\cF}$ to cut for satisfying the $\cF$-predicate, and it uses the same level-restricted \textsc{LP} solution as a \textsc{$\rho$-SEP} guide to find a set of edges $E_{\rho}$ reducing component weights. The union of these two sets of edges is then used to partition the cluster $H$ into smaller clusters, which are recursively processed.

    \begin{algorithm}[H]
\caption{Main \textsc{HC-$\cF$} Procedure}
\label{alg:main-procedure}
\label{alg:main-hac}
\DontPrintSemicolon
\KwIn{A weighted graph $G=(V,E,c,w)$}
\KwOut{A clustering tree for $G$}
\tcp{Preprocessing}
Solve \textsc{LP-HC-$\cF$} on $G$ and obtain a fractional solution $(\mathbf{x},\mathbf{y})$\;
\ForEach{$t\in [w\br{G}]$ and $i\in V$}{
    Set $\hat{x}_i^t \gets 1$ if $x_i^t \geq 1/2$, and set $\hat{x}_i^t \gets 0$ otherwise\;
}
\ForEach{$t\in [w\br{G}]$ and $uv\in E$}{
    $\hat{y}_{uv}^t\gets 2\cdot y_{uv}^t$\;
}
\tcp{Clustering phase}
\Return{\textsc{Recursive-HC-$\cF$}($G,\hat{\mathbf{x}},\hat{\mathbf{y}})$}\;
\end{algorithm}

    \begin{algorithm}[H]
\caption{Recursive \textsc{HC-$\cF$} Procedure}
\label{alg:recursive-procedure}
\label{alg:recursive-hac}
\DontPrintSemicolon
\KwIn{A connected subgraph $H$, together with the guides $(\hat{\mathbf{x}},\hat{\mathbf{y}})$ restricted to $V(H)$}
\KwOut{A clustering tree for $H$}
\If{$H\in \cF$}{
    \Return{a clustering tree $T$ consisting of a single node labeled by $H$}\;
}
Set $r\gets w\br{V\br{H}}$ and $t\gets r/2$\;
Set $X_1\gets \brc{i\in V\br{H}: \hat{x}_i^t = 1}$\;
Run the \textsc{$p_{\cF}$-P} rounding algorithm on \textsc{LP-$p_{\cF}$-P} guide on $H$ at level $t$ from $(\hat{\mathbf{x}},\hat{\mathbf{y}})$ with $X=X_1$ and let $E_{\cF}\subseteq E(H)$ be the returned cut\;

Set $X_0\gets \brc{i\in V\br{H}: \hat{x}_i^t = 0}$\;
Run the weighted \textsc{$\rho$-SEP} rounding algorithm with $\rho=1/2$ and $\rho_0=2/3$ on \textsc{LP-$\rho$-SEP} guide on $H$ at level $t$ from $\hat{\mathbf{y}}$ with $X=X_0$ and let $E_{\rho}\subseteq E(H)$ be the returned cut\;

$E_H\gets E_{\cF}\cup E_{\rho}$\;

Initialize a root node $u_H$ of the output tree $T$ and label $u_H$ by $E_H$\;
\ForEach{$A\in H\setminus E_H$}{
    Recursively compute $T_A\gets $\textsc{Recursive-HC-$\cF$}($A,\hat{\mathbf{x}},\hat{\mathbf{y}})$\;
        Attach $T_A$ as a child of $u_H$ in $T$
}
\Return{$T$}
\end{algorithm}

    \begin{theorem}\label{thm:main-approximation}
        If $\cF$ is a $\alpha\br{n}$-well-behaved class of graphs, then there exists an $\bigo\br{\alpha\br{n}+\beta\br{n}}$-approximation algorithm for the \textsc{HC-$\cF$} problem.
    \end{theorem}
    \begin{proof}
        It is straightforward to see that the algorithm runs in polynomial time and returns a valid $\cF$-clustering tree. It remains to bound the cost of the solution. We start with the following lemma, which shows that each cluster produced by the recursive procedure is either in $\cF$ or has weight at most $2/3$ of its parent cluster. This allows us to charge the cost of each cluster to levels of the \textsc{LP} solution that are not used by any of its children, and hence bound the total cost of the solution by a constant factor times the cost of the \textsc{LP} solution.
    \begin{lemma}\label{lem:child-cluster-valid-or-small}
        Fix a recursive call on a cluster $H\notin\cF$, let $t=w\br{V\br{H}}/2$, let $E_{\cF},E_{\rho}$ be the cuts returned by the two rounding steps in Algorithm~\ref{alg:recursive-procedure}, and define $E_H:=E_{\cF}\cup E_{\rho}$. Then every connected component $A$ of $H\setminus E_H$ satisfies at least one of the following:
        \begin{enumerate}
            \item $A\in \cF$,
            \item $w\br{V\br{A}}\leq 2w\br{V\br{H}}/3$.
        \end{enumerate}
    \end{lemma}
    \begin{proof}
        Recall that $X_1=\brc{i\in V\br{H}\colon \hat x_i^t=1}$ and $X_0=\brc{i\in V\br{H}\colon \hat x_i^t=0}$.

        By Lemma~\ref{lem:lp-fes}, the guide $LP_t^H$ is feasible for \textsc{$p_{\cF}$-P} on terminal set $X_1$. Hence, by the $\alpha\br{n}$-well-behaved rounding guarantee, the removal of $E_{\cF}$ gives that every vertex of $X_1$ satisfies the $\cF$-predicate.

        Now consider the additional deletion of $E_{\rho}$. Since in a well-behaved formulation the predicate is encoded only via $p_{\cF}$-P constraints, deleting extra edges cannot destroy already-satisfied constraints. Therefore every vertex in $X_1$ still satisfies the $\cF$-predicate in $H\setminus E_H$.

        Fix a connected component $A$ of $H\setminus E_H$. If $A\cap X_0=\emptyset$, then $A\subseteq X_1$, so every vertex of $A$ satisfies the $\cF$-predicate in $A$. By definition of a $\cF$-predicate, this implies $A\in\cF$.

        Otherwise, $A\cap X_0\neq\emptyset$. Let $C$ be the connected component of $H\setminus E_{\rho}$ containing $A$. Then $C\cap X_0\neq\emptyset$. By Lemma~\ref{lem:lp-ksep}, $LP_t^H$ is a feasible \textsc{LP-$\rho$-SEP} guide for $X_0$ with $\rho=1/2$, and by the bicriteria guarantee with $\rho_0=2/3$, every connected component of $H\setminus E_{\rho}$ containing a vertex of $X_0$ has weight at most $2w\br{V\br{H}}/3$. Hence $w\br{V\br{A}}\leq w\br{V\br{C}}\leq 2w\br{V\br{H}}/3.$
        This proves the claim.
    \end{proof}

    \begin{lemma}\label{lem:recursive-cost}
        Fix some cluster $H$ for which the recursive procedure was called. Then, 
        \[
        w\br{V\br{H}}\cdot c\br{E_H}\leq 3\br{\alpha\br{n}+\beta\br{n}}\cdot \sum_{t=2w\br{V\br{H}}/3+1}^{w\br{V\br{H}}} c\br{LP_{\fl{t/2}}^H}.
        \]
    \end{lemma}
    \begin{proof}
        By Lemma~\ref{lem:child-cluster-valid-or-small}, every child component $A$ created from $H$ is either already in $\cF$ or has weight at most $2w\br{V\br{H}}/3$. Since recursive calls continue only on non-$\cF$ children, the largest relevant child weight satisfies $r\br{H}\leq 2w\br{V\br{H}}/3$, so levels from $2w\br{V\br{H}}/3+1$ up to $w\br{V\br{H}}$ are available in the charging scheme.

        Firstly, consider the contribution of the edges in $E_{\cF}$ to the cost. Since this cut is done in a cluster of weight $w\br{V\br{H}}$, we have that it is bounded by $w\br{V\br{H}}\cdot c\br{E_{\cF}}$. By applying Lemma \ref{lem:lp-level-bound} with $d=3$ and $b=2$ we get that:
        \begin{align*}
              w\br{V\br{H}}\cdot c\br{E_{\cF}}
              &\leq w\br{V\br{H}}\cdot \alpha\br{n}\cdot c\br{LP_{\fl{w\br{V\br{H}}/2}}^H}\\
              &\leq 3\alpha\br{n}\cdot
              \sum_{t=2w\br{V\br{H}}/3+1}^{w\br{V\br{H}}} c\br{LP_{\fl{t/2}}^H}.
        \end{align*}
        Secondly, the contribution of all $E_{\rho}$ is bounded analogously (using the $\beta\br{n}$-guarantee and Lemma \ref{lem:lp-level-bound} with $d=3$ and $b=2$):
        \[
               w\br{V\br{H}}\cdot c\br{E_{\rho}}
               \leq 3\beta\br{n}\cdot
               \sum_{t=2w\br{V\br{H}}/3+1}^{w\br{V\br{H}}} c\br{LP_{\fl{t/2}}^H}.
        \]
    Since we know that $c\br{E_H}= c\br{E_{\cF}\cup E_{\rho}}\leq c\br{E_{\cF}}+c\br{E_{\rho}}$, the claim follows.
    \end{proof}
    Armed with the above lemma we are ready to sum the contribution of all clusters to the cost of the solution. Let $T$ be the $\cF$-clustering tree returned by the algorithm and let $\cC_T$ be the set of all non-$\cF$-clusters for which the recursive procedure was called. For each such cluster $H$, let $E_H:=E_{\cF}\cup E_{\rho}$ be the set of edges cut in the recursive call on $H$. We have:
    \begin{align*}
        \COST_G\br{T} &= \sum_{H \in \cC_T}\COST_G\br{T, u_H}
        \\&= \sum_{H \in \cC_T} w\br{V\br{H}}\cdot c\br{E_H}\\
        &\leq 3\br{\alpha\br{n}+\beta\br{n}}\cdot\sum_{H \in \cH} \sum_{t=2w\br{V\br{H}}/3+1}^{w\br{V\br{H}}}  c\br{LP_{\fl{t/2}}^H}\\
        &\leq 3\br{\alpha\br{n}+\beta\br{n}}\cdot\sum_{H \in \cH} \sum_{t=r\br{H}+1}^{w\br{V\br{H}}} \br{c\br{LP_{\fl{t/2}}^H}}\\
        &\leq 6\cdot\br{\alpha\br{n}+ \beta\br{n}}\cdot \OPT_{LP}\\
        &\leq 12\cdot\br{\alpha\br{n}+\beta\br{n}}\cdot \OPT_{ILP}.
    \end{align*}
    where the first inequality is due to Lemma \ref{lem:recursive-cost}, the second inequality follows by plugging $r\br{H}\leq 2w\br{V\br{H}}/3$, the third inequality is due to Lemma \ref{lem:relaxation-decomposition} applied with $g=2$, and the last inequality follows from Lemma \ref{lem:lp-rounding}. 
    \end{proof}

\section{Applying the framework}\label{sec:applying-framework}

Now, we showcase that the graph classes of interest are indeed $\alpha\br{n}$-well-behaved for suitable $\alpha\br{n}$. To do so, we invoke some well-known LP-based approximation algorithms for the \textsc{FES} and \textsc{MMC} problems and we show that they certify the well-behavedness of the relevant classes. 
We start by invoking a well-known \textsc{LP} relaxation \textsc{LP-FES} of the \textsc{FES} problem:
    \begin{align}
        \min \quad &\sum_{ij\in E} c\br{i,j} \cdot y_{ij}\\
        \text{subject to} \quad &\sum_{ij\in C} y_{ij} \geq 1, \quad \forall \text{ cycles } C \text{ going through vertex } i\in X,\\
        &0\leq y_{ij} \leq 1, \quad \forall ij\in E.
    \end{align}
    The authors of \cite{ApproximatingMinimumFeedbackSetsAndMulticutsInDirectedGraphs} show that there exists a rounding procedure for this \textsc{LP} with an $\bigo\br{\log n\cdot\log\log n}$ loss of quality:
    \begin{theorem}[\cite{ApproximatingMinimumFeedbackSetsAndMulticutsInDirectedGraphs}]\label{thm:fes-rounding}
        There exists a polynomial time algorithm which given an instance of the \textsc{FES} problem and a solution to the above \textsc{LP-FES}, finds $E'\subseteq E$ such that $G\setminus E'$ contains no cycles going through any vertex in $X$ and $\sum_{ij\in E'} c\br{i,j} \leq \bigo\br{\log n \cdot \log\log n}\cdot \sum_{ij\in E} c\br{i,j} \cdot y_{ij}$.
    \end{theorem}

    Note that originally, this theorem is stated for the more general problem of finding an feedback arc set of a directed graph. For undirected graphs, constant factor approximation algorithm, also based on mathematical programming is known \cite{ConstantFactorApproximationForSubsetFeedbackSetProblemsViaANewLPRelaxation}. However, the ILP formulation used is not $\alpha\br{n}$-well-behaved, and thus it cannot be incorporated in our framework\footnote{Specifically, the issue is that this ILP formulation relies on the fact that the size of $\spr{X}$ is known a-priori, which is not the case in our setting since the number of interesting vertices at each level is not known in advance.}.

    Similarly, the following \textsc{LP} relaxation \textsc{LP-MMC} for the \textsc{MMC} is well-studied:
    \begin{align}
        \min \quad &\sum_{ij\in E} c\br{i,j} \cdot y_{ij}\\
        \text{subject to} \quad &\sum_{ij\in P} y_{ij} \geq 1, \quad \forall \text{ paths } P \text{ from } s_i \text{ to } t_i, i\in [k],\\
        &0\leq y_{ij} \leq 1, \quad \forall ij\in E.
    \end{align}
    In \cite{ApproximateMaxFlowMinMultiCut1993} an algorithm with $\bigo\br{\log n}$-approximation loss using this LP is given:
    \begin{theorem}[\cite{ApproximateMaxFlowMinMultiCut1993}]\label{thm:mmc-rounding}
        There exists a polynomial time algorithm which given an instance of the \textsc{MMC} problem and a solution to the above \textsc{LP-MMC}, finds $E'\subseteq E$ such that $G\setminus E'$ contains no path from $s_i$ to $t_i$ for any $i\in [k]$ and $\sum_{ij\in E'} c\br{i,j} \leq \bigo\br{\log k}\cdot \sum_{ij\in E} c\br{i,j} \cdot y_{ij}$.
    \end{theorem}
    As a corollary of this theorem we obtain that:
    \begin{corollary}
        There exists a polynomial time algorithm which given an instance of the \textsc{P-$\cD_d$} problem and a solution to the above \textsc{LP-P-$\cD_d$}, finds $E'\subseteq E$ such that $G\setminus E'$ is a valid solution for \textsc{P-$\cD_d$} and $\sum_{ij\in E'} c\br{i,j} \leq \bigo\br{\log n}\cdot \sum_{ij\in E} c\br{i,j} \cdot y_{ij}$.
    \end{corollary}

    We get the following well-behavedness corollaries for the three classes of interest:
    \begin{corollary}\label{cor:tree-well-behaved}
        The class $\cT$ of trees is $\bigo\br{\log n\cdot\log\log n}$-well-behaved.
    \end{corollary}
    \begin{proof}
        This follows directly from Theorem~\ref{thm:fes-rounding} and the definition of well-behavedness.
    \end{proof}

    \begin{corollary}\label{cor:diameter-well-behaved}
        For every fixed $d\in\mathbb{R}_{\geq 0}$, the class $\cD_d$ of graphs of diameter at most $d$ is $\bigo\br{\log n}$-well-behaved.
    \end{corollary}
    \begin{proof}
        This follows from Theorem~\ref{thm:mmc-rounding} and the above corollary for \textsc{P-$\cD_d$}.
    \end{proof}


    As immediate consequences of Theorem~\ref{thm:main-approximation}, we obtain approximation guarantees for the corresponding \textsc{HC-$\cF$} problems.

    \begin{corollary}
        There exists an $\bigo\br{\log n \cdot \log\log n}$-approximation algorithm for the \textsc{HC-$\cT$} problem.
    \end{corollary}
    \begin{proof}
        Combine Corollary~\ref{cor:tree-well-behaved} with Theorem~\ref{thm:main-approximation}.
    \end{proof}

    \begin{corollary}
        There exists an $\bigo\br{\log n}$-approximation algorithm for the \textsc{HC-$\cD_d$} problem.
    \end{corollary}
    \begin{proof}
        Combine Corollary~\ref{cor:diameter-well-behaved} with Theorem~\ref{thm:main-approximation}.
    \end{proof}

\section{Inapproximability Results}
In this section, we show that approximating \textsc{HC-$\cT$} and \textsc{HC-$\cD_d$} is essentially as hard as approximating \textsc{HC}. Since the latter problem cannot be approximated within any constant factor under the \textsc{Small Set Expansion Hypothesis} (\textsc{SSEH})~\cite{GraphExpansionAndTheUniqueGamesConjecture}, we obtain the same inapproximability results for \textsc{HC-$\cT$} and \textsc{HC-$\cD_d$}. From now on, we will assume that each graph $G$ has uniform edge costs and vertex weights.

\begin{theorem}[Inapproximability of \textsc{HC} \cite{ApproximateHierarchicalClusteringViaSparsestCutAndSpreadingMetrics}]\label{thm:hc-inapproximability}
    For every $\epsilon > 0$, it is \textsc{SSEH}-hard to distinguish between the following two cases for a given (uniform cost) instance $G$ of \textsc{HC}:
    \begin{itemize}
        \item \textbf{YES case:} There exists a clustering tree $T$ such that $\COST_G\br{T}\leq \epsilon \cdot n \cdot m$.
        \item \textbf{NO case:} For every linear arrangement $\pi$ of $V\br{G}$, $\COST_G\br{\pi}\geq c\sqrt{\epsilon} \cdot n \cdot m$, for some constant $c > 0$.
    \end{itemize}
\end{theorem}
Note that since the objective of the
 \textsc{MLA} is a lower bound on the cost of any clustering tree \cite{ApproximateHierarchicalClusteringViaSparsestCutAndSpreadingMetrics}, the above result implies that in the \textbf{NO} case, every clustering tree $T$ has $\COST_G\br{T}\geq c\sqrt{\epsilon} \cdot n \cdot m$.

\subsection{Inapproximability of \textsc{HC-$\cD_d$}}

We start with the inapproximability of \textsc{HC-$\cD_d$}. For bounded diameter graphs, the proof is simpler than for trees. Here, we assume that the distance between any two vertices in a cluster is measured with respect to the number of edges in the original graph $G$\footnote{For arbitrary metrics it is easy to see that for any value of $d$, the problem is equivalent to \textsc{HC} when we set the distance between all vertices to be $d+\epsilon$, for any $\epsilon\geq 0$.}. Hence, in what follows we may assume that $d\in \mathbb{N}$.
In what follows, we show that by adding appropriate gadgets of diameter $d$ and large degree to every vertex of the input graph $G$, we can basically force any \textsc{HC-$\cD_d$} algorithm to treat such a gadget as a representative of the original vertex. 

\begin{theorem}[Inapproximability of \textsc{HC-$\cD_d$}]\label{thm:hac-d-inapproximability}
    For every $\delta > 0$, it is \textsc{SSEH}-hard to distinguish between the following two cases for a given (uniform cost) instance $H$ of \textsc{HC-$\cD_d$} for any given $d\geq 0$:
    \begin{itemize}
        \item \textbf{YES case:} There exists a $\cD_d$-clustering tree $T$ such that $\COST_H\br{T}\leq \delta \cdot f\br{ n, m, d}$.
        \item \textbf{NO case:} For every $\cD_d$-clustering tree $T$, $\COST_H\br{T}\geq c'\sqrt{\delta} \cdot f\br{ n, m, d}$, for some constant $c' > 0$.
    \end{itemize}
    For some function $f\br{n, m, d}$ (polynomial in $n$, $m$ and $d$).
\end{theorem}
\begin{proof}
    If $d=0$ then the problem is equivalent to \textsc{HC} and the result follows from Theorem~\ref{thm:hc-inapproximability}. Assume that $d\geq 1$. We create a \textsc{HC-$\cD_d$} instance $H$ by modifying $G$ as follows: For any vertex $v\in V\br{G}$, $i\in [n\br{G}]$ and $j\in [d]$ we create a vertex $u_{i,j}^v$. Then, we connect $u_{i,j}^v$ to $u_{i',j'}^v$ if $\spr{j-j'}=1$. We also connect each $u_{i, 1}^v$ to $v$ for every $i\in [n\br{G}]$. We denote the subgraph consisting of aforementioned representatives of $v$ as $H^v$ and note that it has diameter $d$, its minimum cut is of size $n\br{G}$ and $n\br{H^v}=n\br{G}\cdot d +1$.
    Figure~\ref{fig:hardness-diameter-reduction} illustrates the reduction on a small example and shows both a single gadget $H^v$ and the full graph obtained after attaching the gadgets to all vertices of $G$.

    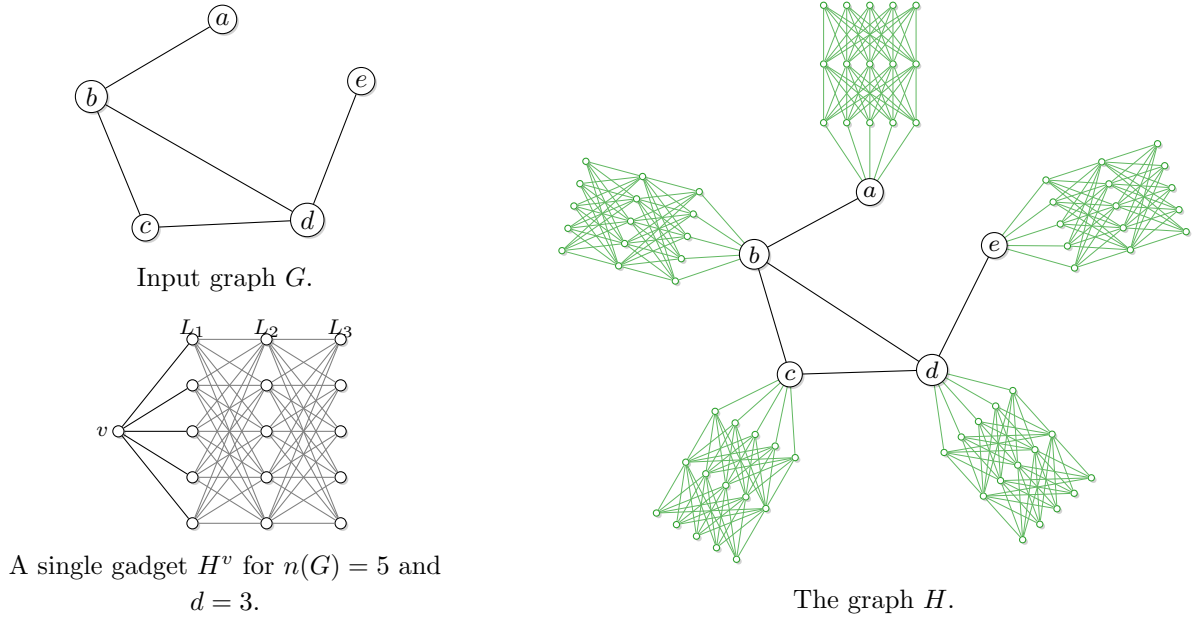
\begin{figure}[H]
    \centering
    \begin{minipage}[c][9.5cm][c]{0.38\textwidth}
        \centering
        \begin{tikzpicture}[scale=1.02,every node/.style={circle,draw,fill=white,inner sep=1.4pt,font=\small,preaction={fill=black,opacity=0.18,transform canvas={xshift=0.8pt,yshift=-0.8pt}}}]
            \node (a) at (0.0,1.4) {$a$};
            \node (b) at (-1.7,0.4) {$b$};
            \node (c) at (-1.0,-1.3) {$c$};
            \node (d) at (1.1,-1.2) {$d$};
            \node (e) at (1.8,0.6) {$e$};

            \draw (a) -- (b);
            \draw (b) -- (c);
            \draw (c) -- (d);
            \draw (d) -- (e);
            \draw (b) -- (d);
        \end{tikzpicture}
        \\\smallskip
        {\small Input graph $G$.}
        \\\medskip
        \begin{tikzpicture}[
            x=0.98cm,
            y=0.76cm,
            vertex/.style={circle,draw,fill=white,inner sep=0pt,minimum size=4.2pt,preaction={fill=black,opacity=0.18,transform canvas={xshift=0.8pt,yshift=-0.8pt}}},
            note/.style={draw=none,font=\scriptsize}
        ]
            \node[vertex] (v) at (0,0) {};
            \node[note,left] at (v) {$v$};

            \node[note] at (1,1.8) {$L_1$};
            \node[note] at (2,1.8) {$L_2$};
            \node[note] at (3,1.8) {$L_3$};

            \foreach \i/\y in {1/1.6,2/0.8,3/0,4/-0.8,5/-1.6} {
                \node[vertex] (a\i) at (1,\y) {};
                \node[vertex] (b\i) at (2,\y) {};
                \node[vertex] (c\i) at (3,\y) {};
            }

            \foreach \i in {1,...,5} {
                \draw (v) -- (a\i);
            }
            \foreach \i in {1,...,5} {
                \foreach \j in {1,...,5} {
                    \draw[gray] (a\i) -- (b\j);
                    \draw[gray] (b\i) -- (c\j);
                }
            }
        \end{tikzpicture}
        \\\smallskip
        {\small A single gadget $H^v$ for $n(G)=5$ and $d=3$.}
    \end{minipage}
    \hfill
    \begin{minipage}[c][9.5cm][c]{0.58\textwidth}
        \centering
        \begin{tikzpicture}[
            x=1.05cm,
            y=1.05cm,
            scale=1.12,
            transform shape,
            vertex/.style={circle,draw,fill=white,inner sep=1.2pt,font=\scriptsize,preaction={fill=black,opacity=0.18,transform canvas={xshift=0.8pt,yshift=-0.8pt}}},
            gadget/.style={circle,draw,inner sep=0pt,minimum size=2.1pt,draw=cutC,fill=white,preaction={fill=black,opacity=0.18,transform canvas={xshift=0.8pt,yshift=-0.8pt}}},
            gadgetedge/.style={draw=cutC!75,line width=0.22pt}
        ]
            \coordinate (a) at (0.0,1.1);
            \coordinate (b) at (-1.3,0.4);
            \coordinate (c) at (-0.9,-0.95);
            \coordinate (d) at (0.7,-0.9);
            \coordinate (e) at (1.4,0.5);

            \draw (a) -- (b);
            \draw (b) -- (c);
            \draw (c) -- (d);
            \draw (d) -- (e);
            \draw (b) -- (d);

            \foreach \name/\ang in {a/90,b/165,c/240,d/312,e/18} {
                \begin{scope}[shift={(\name)}, rotate=\ang]
                    \foreach \i/\yy in {1/0.52,2/0.26,3/0,4/-0.26,5/-0.52} {
                        \node[gadget] (\name1\i) at (0.78,\yy) {};
                        \node[gadget] (\name2\i) at (1.44,\yy) {};
                        \node[gadget] (\name3\i) at (2.1,\yy) {};
                    }
                    \foreach \i in {1,...,5} {
                        \draw[gadgetedge] (0,0) -- (\name1\i);
                    }
                    \foreach \i in {1,...,5} {
                        \foreach \j in {1,...,5} {
                            \draw[gadgetedge] (\name1\i) -- (\name2\j);
                            \draw[gadgetedge] (\name2\i) -- (\name3\j);
                        }
                    }
                \end{scope}
            }

            \node[vertex] at (a) {$a$};
            \node[vertex] at (b) {$b$};
            \node[vertex] at (c) {$c$};
            \node[vertex] at (d) {$d$};
            \node[vertex] at (e) {$e$};
        \end{tikzpicture}
        \\\smallskip
        {\small The graph $H$.}
    \end{minipage}
    \caption{Reduction from Theorem~\ref{thm:hac-d-inapproximability} for $n(G)=5$ and $d=3$. The left side shows the input graph and a single gadget $H^v$, while the right side shows the full graph $H$.}
    \label{fig:hardness-diameter-reduction}
\end{figure}

    Firstly, assume that $G$ is a \textbf{YES} instance and $T$ is a clustering tree such that $\COST_G\br{T}\leq \epsilon \cdot n\br{G} \cdot m\br{G}$. We construct a $\cD_d$-clustering tree $T'$ for $H$ as follows. We start with $T'=T$ and we modify all clusters in the following way: for every cluster $A\in \cC^{T'}$ and a vertex $v\in A$ we add all the vertices of $H^v$ to $A$ as well. This transformation increases the cost of $T'$ by a multiplicative factor of $n\br{G}\cdot d+1$, ensures that $T'$ is a $\cD_d$-clustering tree. In particular all of the leaf clusters in $T'$ have diameter at most $d$. We have that \[
    \COST_H\br{T'} \leq \br{n\br{G}\cdot d+1}\cdot\COST_G\br{T}\leq \epsilon \cdot \br{n\br{G}\cdot d+1}\cdot n\br{G}\cdot m\br{G}.
    \]

    Secondly, assume that $G$ is a \textbf{NO} instance. Observe, that we can assume that an optimal $\cD_d$-clustering tree never cuts edges of $H^v$. If otherwise, let $\cC$ be the set of edges cut in $H^v$. There are two cases:
    \begin{enumerate}
        \item $\spr{H^v\setminus \cC}= 1$. Then, we can remove cutting of the edges in $\cC$ from $T$ without affecting the structure of clusters. Since the cluster containing $V\br{H^v}$ had diameter at most $d$ before, then it must have only consisted of vertices of $H^v$ and thus, after transformation it still has diameter at most $d$.
        \item $\spr{H^v\setminus \cC}\geq 2$. 
        Let $\cS\subseteq \cC$ be the set of edges such that $\cS$ consists of all of the edges cut on the path from root of $T$ towards the cluster containing $V\br{H^v}$ before it was first cut into at least two clusters. Let $C$ be the first cluster in this top-down exploration such that $\spr{C\setminus \cC}\geq 2$. Now, iteratively append edges from $\cC\setminus\cS$, to $\cS$ as long as $\spr{H^v\setminus \cS}=1$. 
        Observe that 
        \[
        \spr{\cS}> n\br{G} - 1 \geq \deg_G\br{v}
        \]
        since the minimum cut of $H^v$ is of size $n\br{G}$.
         By definition $\cS$ is not a cut of $H^v$, so we can remove cutting of these edges from clusters of $T$ without affecting the structure of the tree, similar as in the previous case. We replace them by cutting all of the (remaining) edges of $N_G\br{v}$ in $C$. Observe, that there are at most $\deg_G\br{v}<\spr{\cS}$ such edges, and thus, the cost of $T$ is strictly less then at the beginning. Moreover, upon doing so, all clusters containing vertices $V\br{H^v}$ become $\cD_d$-clusters and therefore, we can remove their child clusters from $T$. Finally, it is easy to see that each other child cluster of $C$ is a subgraph of some child of $C$ in the original tree, so $T$ is still a valid $\cD_d$-clustering tree.
        \end{enumerate}

    Therefore, the only edges which can be cut in any sensible $\cD_d$-clustering tree are the edges of $G$. However, in such case the cost of any $\cD_d$-clustering tree $T'$ obtained by repeating the cuts of the clustering tree $T$ for $G$ is lower-bounded by
    \[
    \COST_H\br{T'} = \br{n\br{G}\cdot d+1}\cdot\COST_G\br{T} \geq c\sqrt{\epsilon} \cdot \br{n\br{G}\cdot d+1}\cdot n\br{G}\cdot m\br{G}.
    \]
        where we reiterate that the multiplicative factor of $n\br{G}\cdot d+1$ comes from the fact that every edge cut in $G$ contributes to the cost of every vertex in the corresponding gadget $H^v$.

    By choosing $f\br{n, m, d} = \br{n\br{G}\cdot d+1}\cdot n\br{G}\cdot m\br{G}$, $\delta = \epsilon$ and $c'=c$, we obtain the result.
\end{proof}

\subsection{Inapproximability of \textsc{HC-$\cT$}}

Here, in order to prove the inapproximability of \textsc{HC-$\cT$}, we will use a similar approach as in the previous section. However, now the gadgets we use serve a different purpose. In order for the reduction to work, we need to use the fact that the cost of the optimal \textsc{MLA} is a lower bound on the cost of any clustering tree. However, this might not be a case for instances of \textsc{HC-$\cT$}. The purpose of gadgets we use is to ensure that a similar lower bound can also be established for the graphs obtained using the reduction.
\begin{theorem}[Inapproximability of \textsc{HC-$\cT$}]\label{thm:hac-inapproximability}
    For every $\delta > 0$, it is \textsc{SSEH}-hard to distinguish between the following two cases for a given (uniform cost) instance $H$ of \textsc{HC-$\cT$}:
    \begin{itemize}
        \item \textbf{YES case:} There exists a $\cT$-clustering tree $T$ such that $\COST_H\br{T}\leq \delta \cdot n \cdot m$.
        \item \textbf{NO case:} Every $\cT$-clustering tree $T$, $\COST_H\br{T}\geq c'\sqrt{\delta} \cdot n \cdot m$, for some constant $c' > 0$.
    \end{itemize}
\end{theorem}
\begin{proof}
    Let $G=\br{V,E}$ be an instance of \textsc{HC}. We set $\delta = 2\epsilon$ and $c'= \frac{c}{9\sqrt{2}}$. We create a \textsc{HC-$\cT$} instance $H$ as follows: For any vertex $v\in V\br{G}$ we add two new vertices $v_1$ and $v_2$ and we build a triangle on $\brc{v,v_1,v_2}$. 
    Figure~\ref{fig:hardness-tree-reduction} shows the construction on a small example.

    \begin{figure}[H]
    \centering
    \begin{minipage}[c][8.3cm][c]{0.28\textwidth}
        \centering
        \begin{tikzpicture}[scale=1.12,every node/.style={circle,draw,fill=white,inner sep=1.4pt,font=\small,preaction={fill=black,opacity=0.18,transform canvas={xshift=0.8pt,yshift=-0.8pt}}}]
            \node (a) at (0.0,1.4) {$a$};
            \node (b) at (-1.7,0.4) {$b$};
            \node (c) at (-1.0,-1.3) {$c$};
            \node (d) at (1.1,-1.2) {$d$};
            \node (e) at (1.8,0.6) {$e$};

            \draw (a) -- (b);
            \draw (a) -- (c);
            \draw (b) -- (c);
            \draw (c) -- (d);
            \draw (b) -- (d);
            \draw (d) -- (e);
        \end{tikzpicture}
        \\\smallskip
        {\small Input graph $G$.}
    \end{minipage}
    \hfill
    \begin{minipage}[c][8.3cm][c]{0.68\textwidth}
        \centering
        \begin{tikzpicture}[
            x=1.15cm,
            y=1.15cm,
            every node/.style={circle,draw,fill=white,inner sep=1.2pt,font=\scriptsize,preaction={fill=black,opacity=0.18,transform canvas={xshift=0.8pt,yshift=-0.8pt}}},
            base/.style={circle,draw,fill=white,inner sep=1.2pt,font=\scriptsize,preaction={fill=black,opacity=0.18,transform canvas={xshift=0.8pt,yshift=-0.8pt}}},
            tnode/.style={circle,draw=cutD,fill=white,inner sep=1.1pt,font=\scriptsize,transform shape=false,preaction={fill=black,opacity=0.18,transform canvas={xshift=0.8pt,yshift=-0.8pt}}},
            tedge/.style={draw=cutD,very thick}
        ]
            \coordinate (a) at (0.0,1.2);
            \coordinate (b) at (-1.6,0.45);
            \coordinate (c) at (-1.0,-1.15);
            \coordinate (d) at (0.95,-1.05);
            \coordinate (e) at (1.7,0.5);

            \draw (a) -- (b);
            \draw (a) -- (c);
            \draw (b) -- (c);
            \draw (c) -- (d);
            \draw (b) -- (d);
            \draw (d) -- (e);

            \begin{scope}[shift={(a)}, rotate=90]
                \node[tnode] (a1) at (0.866,0.5) {$a_1$};
                \node[tnode] (a2) at (0.866,-0.5) {$a_2$};
                \draw[tedge] (0,0) -- (a1);
                \draw[tedge] (0,0) -- (a2);
                \draw[tedge] (a1) -- (a2);
            \end{scope}

            \begin{scope}[shift={(b)}, rotate=165]
                \node[tnode] (b1) at (0.866,0.5) {$b_1$};
                \node[tnode] (b2) at (0.866,-0.5) {$b_2$};
                \draw[tedge] (0,0) -- (b1);
                \draw[tedge] (0,0) -- (b2);
                \draw[tedge] (b1) -- (b2);
            \end{scope}

            \begin{scope}[shift={(c)}, rotate=240]
                \node[tnode] (c1) at (0.866,0.5) {$c_1$};
                \node[tnode] (c2) at (0.866,-0.5) {$c_2$};
                \draw[tedge] (0,0) -- (c1);
                \draw[tedge] (0,0) -- (c2);
                \draw[tedge] (c1) -- (c2);
            \end{scope}

            \begin{scope}[shift={(d)}, rotate=312]
                \node[tnode] (d1) at (0.866,0.5) {$d_1$};
                \node[tnode] (d2) at (0.866,-0.5) {$d_2$};
                \draw[tedge] (0,0) -- (d1);
                \draw[tedge] (0,0) -- (d2);
                \draw[tedge] (d1) -- (d2);
            \end{scope}

            \begin{scope}[shift={(e)}, rotate=18]
                \node[tnode] (e1) at (0.866,0.5) {$e_1$};
                \node[tnode] (e2) at (0.866,-0.5) {$e_2$};
                \draw[tedge] (0,0) -- (e1);
                \draw[tedge] (0,0) -- (e2);
                \draw[tedge] (e1) -- (e2);
            \end{scope}

            \node[base] at (a) {$a$};
            \node[base] at (b) {$b$};
            \node[base] at (c) {$c$};
            \node[base] at (d) {$d$};
            \node[base] at (e) {$e$};
        \end{tikzpicture}
        \\\smallskip
        {\small The graph $H$ with triangles attached to each vertex of $G$.}
    \end{minipage}
    \caption{Illustration of the reduction used in Theorem~\ref{thm:hac-inapproximability}. For each original vertex $v\in V(G)$ we add two new vertices $v_1$ and $v_2$ and create the triangle on $\brc{v,v_1,v_2}$.}
    \label{fig:hardness-tree-reduction}
\end{figure}
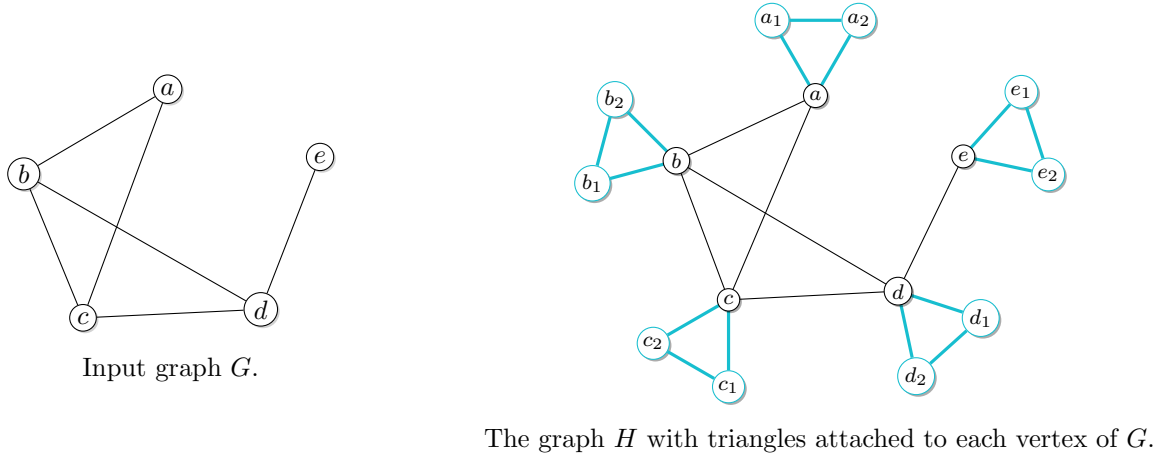
    
    Firstly, assume that $G$ is a \textbf{YES} instance and $T$ is a clustering tree such that $\COST_G\br{T}\leq \epsilon \cdot n\br{G} \cdot m\br{G}$. Then, we show how to construct a good $\cT$-clustering tree $T'$ for $H$: We start with $T'=T$ and we modify all clusters in the following way: for every cluster $A\in \cC^{T'}$ and vertex $v\in A$ we add $v_1$ and $v_2$ to $A$ as well. This transformation increases the cost of $T'$ be a multiplicative factor of $3$, ensures that all of the clusters in $T'$ satisfy hierarchical clustering properties except that every leaf cluster in $T'$ is a triangle. To fix this, we add a new level of the tree, in which we cut each such triangle using one, arbitrarily chosen edge. It is easy to see that
    \[
    \COST_H\br{T'} = 3\COST_G\br{T}+ 3n\br{G} \leq 3\epsilon \cdot n\br{G} \cdot m\br{G} + 3n\br{G} \leq \epsilon \cdot n\br{H} \cdot m\br{H}+n\br{H}\leq \delta \cdot n\br{H} \cdot m\br{H}
    \]
    where we used the fact that $m\br{G}\leq m\br{H}$ and $\epsilon \geq 1/m\br{H}$ for large enough $m\br{H}$.
    
    Secondly, assume that $G$ is a \textbf{NO} instance. We want to show that the cost of any $\cT$-clustering tree $T$ for $H$ is lower bounded by the cost of an optimal linear arrangement of $G$. To do so, fix some clustering tree $T$ for $G$, pick all original vertices of $G$ and project them onto a line according to the order in which they occur in the leaf clusters of $T$, thus obtaining a linear arrangement of $G$, $\pi$ (assume that both the tree and vertices of each cluster have some arbitrary order). Consider an edge $uv \in E$. There are two cases:
    \begin{enumerate}
    \item $u$ and $v$ belong to different leaf clusters in $T$. Then, we know that the edge $uv$ had to be cut in $T$ and the stretch of $u$ and $v$ in $\pi$, i.~e., $\spr{\pi\br{u}-\pi\br{v}}$, is at most the size of the cluster $C$ in which $uv$ was cut. On the other hand, the contribution of $uv$ to $c_G\br{T}$ is $\spr{\cC}\geq \spr{\pi\br{u}-\pi\br{v}}$.
    \item $u$ and $v$ belong to the same leaf cluster in $T$. Consider the cluster $C$ in which $u$ and $v$ belong. Let $C'$ be this cluster restricted to the vertices of $G$. Since $C'$ is a tree, we know that the contribution of edges of $C'$ to the cost of $\pi$ is at most $\br{\spr{C'}-1}\cdot \spr{C'}\leq \spr{C'}^2$. On the other hand, we know that for each vertex $v\in C'$, at least one of the edges $vv_1, vv_2$ has to be cut in $T$ and every such edge is cut in a cluster of size at least $\spr{C'}$ so the contribution of these edges to $c_H\br{T}$ is at lower bounded by $\spr{C'}^2$. Therefore, the contribution of edges of $C$ to the cost of $T$ is at least the contribution of edges of $C'$ to the cost of $\pi$. Figure~\ref{fig:hardness-tree-no-case2} illustrates the above argument on a small example.
    \end{enumerate}   

    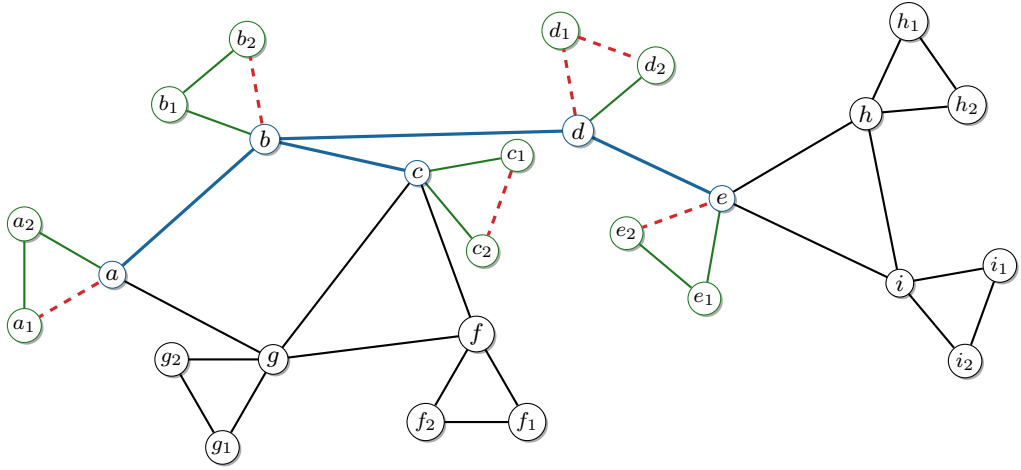
\begin{figure}[H]
    \centering
    \begin{tikzpicture}[
        scale=1.12,
        transform shape,
        every node/.style={circle,draw,fill=white,inner sep=1.2pt,font=\scriptsize,preaction={fill=black,opacity=0.18,transform canvas={xshift=0.8pt,yshift=-0.8pt}}},
        cvert/.style={circle,draw=cutC!70!black,fill=white,inner sep=1.2pt,font=\scriptsize,preaction={fill=black,opacity=0.18,transform canvas={xshift=0.8pt,yshift=-0.8pt}}},
        cpvert/.style={circle,draw=cutB!80!black,fill=white,inner sep=1.2pt,font=\scriptsize,preaction={fill=black,opacity=0.18,transform canvas={xshift=0.8pt,yshift=-0.8pt}}},
        outvert/.style={circle,draw=black,fill=white,inner sep=1.2pt,font=\scriptsize,preaction={fill=black,opacity=0.18,transform canvas={xshift=0.8pt,yshift=-0.8pt}}},
        aux/.style={circle,draw=cutC!75!black,fill=white,inner sep=1.1pt,font=\scriptsize,transform shape=false,preaction={fill=black,opacity=0.18,transform canvas={xshift=0.8pt,yshift=-0.8pt}}},
        auxblack/.style={circle,draw=black,fill=white,inner sep=1.1pt,font=\scriptsize,transform shape=false,preaction={fill=black,opacity=0.18,transform canvas={xshift=0.8pt,yshift=-0.8pt}}},
        cedge/.style={draw=cutC!75!black,thick},
        cpedge/.style={draw=cutB!85!black,very thick},
        outedge/.style={draw=black,thick},
        cutedge/.style={draw=cutA,very thick,dashed}
    ]
        \coordinate (a) at (-4.0,0.0);
        \coordinate (b) at (-2.2,1.6);
        \coordinate (c) at (-0.4,1.2);
        \coordinate (d) at (1.5,1.7);
        \coordinate (e) at (3.2,0.9);
        \coordinate (f) at (0.3,-0.7);
        \coordinate (g) at (-2.1,-1.0);
        \coordinate (h) at (4.9,1.9);
        \coordinate (i) at (5.3,-0.1);

        \draw[cpedge] (a) -- (b);
        \draw[cpedge] (b) -- (c);
        \draw[cpedge] (b) -- (d);
        \draw[cpedge] (d) -- (e);

        \draw[outedge] (c) -- (f);
        \draw[outedge] (f) -- (g);
        \draw[outedge] (g) -- (a);
        \draw[outedge] (c) -- (g);

        \draw[outedge] (e) -- (h);
        \draw[outedge] (h) -- (i);
        \draw[outedge] (e) -- (i);

        \begin{scope}[shift={(a)}, rotate=180]
            \node[aux] (a1) at (1.039,0.6) {$a_1$};
            \node[aux] (a2) at (1.039,-0.6) {$a_2$};
            \draw[cutedge] (0,0) -- (a1);    
            \draw[cedge] (0,0) -- (a2);
            \draw[cedge] (a1) -- (a2);
        \end{scope}

        \begin{scope}[shift={(b)}, rotate=130]
            \node[aux] (b1) at (1.039,0.6) {$b_1$};
            \node[aux] (b2) at (1.039,-0.6) {$b_2$};
            \draw[cedge] (0,0) -- (b1);
            \draw[cutedge] (0,0) -- (b2);    
            \draw[cedge] (b1) -- (b2);
        \end{scope}

        \begin{scope}[shift={(c)}, rotate=340]
            \node[aux] (c1) at (1.039,0.6) {$c_1$};
            \node[aux] (c2) at (1.039,-0.6) {$c_2$};
            \draw[cedge] (0,0) -- (c1);
            \draw[cedge] (0,0) -- (c2);
            \draw[cutedge] (c1) -- (c2);     
        \end{scope}

        \begin{scope}[shift={(d)}, rotate=70]
            \node[aux] (d1) at (1.039,0.6) {$d_1$};
            \node[aux] (d2) at (1.039,-0.6) {$d_2$};
            \draw[cutedge] (0,0) -- (d1);    
            \draw[cedge] (0,0) -- (d2);
            \draw[cutedge] (d1) -- (d2);
        \end{scope}

        \begin{scope}[shift={(e)}, rotate=230]
            \node[aux] (e1) at (1.039,0.6) {$e_1$};
            \node[aux] (e2) at (1.039,-0.6) {$e_2$};
            \draw[cedge] (0,0) -- (e1);
            \draw[cutedge] (0,0) -- (e2);
            \draw[cedge] (e1) -- (e2);
        \end{scope}

        \begin{scope}[shift={(f)}, rotate=-90]
            \node[auxblack] (f1) at (1.039,0.6) {$f_1$};
            \node[auxblack] (f2) at (1.039,-0.6) {$f_2$};
            \draw[outedge] (0,0) -- (f1);
            \draw[outedge] (0,0) -- (f2);
            \draw[outedge] (f1) -- (f2);
        \end{scope}

        \begin{scope}[shift={(g)}, rotate=-150]
            \node[auxblack] (g1) at (1.039,0.6) {$g_1$};
            \node[auxblack] (g2) at (1.039,-0.6) {$g_2$};
            \draw[outedge] (0,0) -- (g1);
            \draw[outedge] (0,0) -- (g2);
            \draw[outedge] (g1) -- (g2);
        \end{scope}

        \begin{scope}[shift={(h)}, rotate=35]
            \node[auxblack] (h1) at (1.039,0.6) {$h_1$};
            \node[auxblack] (h2) at (1.039,-0.6) {$h_2$};
            \draw[outedge] (0,0) -- (h1);
            \draw[outedge] (0,0) -- (h2);
            \draw[outedge] (h1) -- (h2);
        \end{scope}

        \begin{scope}[shift={(i)}, rotate=-20]
            \node[auxblack] (i1) at (1.039,0.6) {$i_1$};
            \node[auxblack] (i2) at (1.039,-0.6) {$i_2$};
            \draw[outedge] (0,0) -- (i1);
            \draw[outedge] (0,0) -- (i2);
            \draw[outedge] (i1) -- (i2);
        \end{scope}

        \node[cpvert] at (a) {$a$};
        \node[cpvert] at (b) {$b$};
        \node[cpvert] at (c) {$c$};
        \node[cpvert] at (d) {$d$};
        \node[cpvert] at (e) {$e$};
        \node[outvert] at (f) {$f$};
        \node[outvert] at (g) {$g$};
        \node[outvert] at (h) {$h$};
        \node[outvert] at (i) {$i$};

    \end{tikzpicture}
    \caption{Illustration of Case 2 in the \textbf{NO} analysis for \textsc{HC-$\cT$}. Blue denotes $C'$ (original vertices), green denotes vertices and edges of $C\setminus C'$, and red dashed edges are the ones cut in triangles.}
    \label{fig:hardness-tree-no-case2}
\end{figure}
    
    By summing up the contributions of all edges, according to the above two cases, we obtain that 
    \[
    \COST_H\br{T}\geq \COST_G\br{\pi}\geq c\sqrt{\epsilon} \cdot n\br{G} \cdot m\br{G} \geq \frac{c\sqrt{\epsilon}}{9}\cdot n\br{H} \cdot m\br{H} = c'\sqrt{\delta} \cdot n\br{H} \cdot m\br{H}
    \]
    where we used the fact that $n\br{H}=3n\br{G}$ and 
    \[
    m\br{H}= m\br{G}+2\cdot n\br{G}\leq 3\cdot m\br{G}
    \]
    since we can assume that $G$ is not a tree (for which \textsc{HC} is constant-factor approximable), so $n\br{G}\leq m\br{G}$.
\end{proof}

\section{Conclusions}

We have presented a general framework for approximating \textsc{HC-$\cF$} problems for $\cF$ which are $\alpha\br{n}$-well-behaved, i.e., there exists a natural, cover-like ILP formulation of the \textsc{$p_{\cF}$-P} problem as well as an algorithm which given a solution to the LP relaxation of this ILP, finds a solution to \textsc{$p_{\cF}$-P} with cost at most $\alpha\br{n}$ times the cost of the LP solution. This results in an $\bigo\br{\alpha\br{n}+\beta\br{n}}$-approximation algorithm for \textsc{HC-$\cF$}, where $\beta\br{n}$ is the approximation ratio of the bicriteria LP-based algorithm for \textsc{$\rho$-SEP} with $\rho=1/2$ and $\rho_0=2/3$, currently at $\beta\br{n}=\bigo\br{\log n}$. As an immediate corollary we obtain an $\bigo\br{\log n\cdot\log\log n}$-approximation algorithm for clustering into trees and an $\bigo\br{\log n}$-approximation algorithm for clustering into clusters of diameter at most $d$. To the best of our knowledge, these are the first results for the \textsc{HC-$\cF$} problem.

As of right now, improving the approximation of \textsc{HC-$\cF$} beyond the $\bigo\br{\log n}$-factor seems unlikely, without improving the approximation for a series of intensively studied cut and separator problems. Additional evidence for hardness of approximation are given by our inapproximability results, which show that under \textsc{SSEH}, \textsc{HC-$\cD_d$} (for any $d\geq 0$) and \textsc{HC-$\cT$} cannot be approximated within a constant factor.

It is worth mentioning that among corollaries of our work, we can also include an $\bigo\br{\log n}$-approximation algorithm for the following class of graphs $\cF$: given a set of vertex pairs $\{(s_i, t_i)\in V\times V\colon i\in [k]\}$, the graph $G$ is in $\cF$ if and only if for every $i\in [k]$, there is no pair of vertices $s_i$ and $t_i$ in the same connected component of $G$. We call this problem the \textsc{Multiway Cut Hierarchical Clustering} problem. Indeed, while designing an approximation algorithm for $\cD_d$-clustering, we have used the fact that $p_{\cD_d}$-P is a special case of \textsc{Multiway Cut} problem. It is immediate, that by the same argument, the above class is $\bigo\br{\log n}$-well-behaved as well.

\bibliographystyle{plain}
\bibliography{references}

\appendix

\section{Spreading metric for $\rho$-separators with terminals}\label{app:terminal-extension}

This appendix serves as an extension of the procedure and analysis in~\cite{FastApproximateGraphPartitioningAlgorithms} which gives an $\bigo\br{\log n}$-approximation algorithm for the \textsc{$\rho$-SEP} problem. In their algorithm it is assumed that $X=V$, we show to extend the rounding procedure to work for designated terminal set $X\subseteq V$. The algorithm is almost exactly the same, and is based on a region-growing procedure. Whereas the original procedure grows a region around an arbitrary vertex, our modified algorithm grows a region around an arbitrary terminal. The rest of the analysis is largely the same.

\subsection{LP formulation}
Recall, the following is the weighted spreading-metrics \textsc{LP} relaxation \textsc{LP-$\rho$-SEP} for the \textsc{$\rho$-SEP} problem with a terminal set $X$:
\begin{align}
    \min \quad &\sum_{e\in E} c(e)\cdot y_e \\
    \text{subject to} \quad &\sum_{u\in S} \dist_{y,V}(v,u)\cdot w(u) \geq w\br{S}-\rho\cdot w\br{G}, \quad \forall S\subseteq V, v\in S\cap X, \label{eq:px1-main}\\
    &0\leq y_e \quad \forall e\in E.
\end{align}

Again, the LP can be solved in polynomial time using the ellipsoid method, since there is a polynomial-time separation oracle for the constraints \eqref{eq:px1-main}. Let $\tau_X$ denote the optimal value of \textsc{LP-$\rho$-SEP}.

\subsection{Lower bound}
Firstly, we claim that $\tau_X$ is a lower bound on the minimum capacity of a $(\rho,X)$-separator.
\begin{lemma}
$\tau_X$ is a lower bound on the minimum capacity of a $(\rho,X)$-separator.
\end{lemma}

\begin{proof}
Let $F$ be any $(\rho,X)$-separator. Define an integral length assignment by
\[
    y_e=\begin{cases}
        1,& e\in F,\\
        0,& e\notin F.
    \end{cases}
\]
Take any constraint \eqref{eq:px1-main}, i.e., any $S\subseteq V$ and any $v\in S\cap X$.
Let $C$ be the connected component of $v$ in $G'=(V,E\setminus F)$.
For every $u\in S\setminus C$, every $v$-$u$ path uses at least one edge of $F$, hence $\dist_{y,V}(v,u)\ge 1$.
Therefore,
\begin{align*}
\sum_{u\in S}\dist_{y,V}(v,u)\cdot w(u)
&\ge \sum_{u\in S\setminus C}\dist_{y,V}(v,u)\cdot w(u)\\
&\ge w\br{S\setminus C}\\
&= w\br{S}-w\br{C\cap S}\\
&\ge w\br{S}-w\br{C}.
\end{align*}
Since $v\in X$, the component $C$ contains a terminal, so by definition of $(\rho,X)$-separator,
$w(C)\le \rho\cdot w\br{G}$. Hence
\[
\sum_{u\in S}\dist_{y,V}(v,u)\cdot w(u)\ge w\br{S}-\rho\cdot w\br{G},
\]
so every constraint is satisfied. The LP cost equals $c(F)$, thus
$\tau_X\le c(F)$ for every feasible $(\rho,X)$-separator $F$. Consequently,
$\tau_X$ is a lower bound on the optimum separator capacity.
\end{proof}

\subsection{Radius bound}

For a vertex $v\in S\subseteq V$, define $
\operatorname{radius}(v,S)=\max_{u\in S}\brc{\dist_S(v,u)}.$
We have the following lemma:

\begin{lemma}
If $S\subseteq V$ satisfies $w\br{S}>\rho_0\cdot w\br{G}$, then for every terminal $v\in S\cap X$, $
\operatorname{radius}(v,S)>\frac{\rho_0-\rho}{\rho_0}.
$
\end{lemma}

\begin{proof}
Fix $S$ and terminal $v\in S\cap X$. By \eqref{eq:px1-main}, we have $
\sum_{u\in S}\dist_{y,V}(v,u)\cdot w(u)\ge w\br{S}-\rho\cdot w\br{G}.
$
Hence the weighted average distance from $v$ to vertices in $S$ is at least
\[
1-\frac{\rho\cdot w\br{G}}{w\br{S}} > 1-\frac{\rho}{\rho_0}=\frac{\rho_0-\rho}{\rho_0}.
\]
By averaging there exists a vertex $u\in S$ such that $
\dist_{y,V}(v,u)>\frac{\rho_0-\rho}{\rho_0}.
$
Since $\dist_{y,S}(v,u)\ge \dist_{y,V}(v,u)$, we obtain
\[
\operatorname{radius}(v,S)\ge \dist_S(v,u)>\frac{\rho_0-\rho}{\rho_0}.
\]
\end{proof}

\subsection{Assigning volumes to spheres}

Fix a connected subgraph $G'=(V',E')$ that we currently process.
For a vertex $v\in V'$ and radius $r\ge 0$, define an \emph{$r$-sphere} as $
N(v,r)=\{u\in V'\colon \dist_{V'}(u,v)<r\}$. Let $
E(v,r)=E'\cap \bigl(N(v,r)\times N(v,r)\bigr) $ be the set of edges with both endpoints in the $r$-sphere, and let $
\delta(v,r)=E'\cap \bigl(N(v,r),V'\setminus N(v,r)\bigr)$ be the set of edges with exactly one endpoint in the $r$-sphere.

For a parameter $\epsilon>0$, define the volume of a sphere by
\[
\operatorname{vol}(v,r)=\epsilon\cdot \tau_X
+\sum_{e\in E(v,r)} c(e)\cdot y_e
+\sum_{(x,y)\in \delta(v,r)} c(x,y)\cdot y_{x,y}\cdot \br{r-\dist_{N(v,r)}(v,x)}.
\]

This quantity consists of 3 terms: a small constant term $\epsilon\cdot\tau_X$ called \emph{seed value}, the contribution of edges inside the sphere, and the contribution of edges crossing the sphere boundary. We will choose $\epsilon=\frac{1}{n\cdot\ln n}$.
Now, we observe that $\operatorname{vol}(v,r)$ is monotone and piecewise linear. This is because the first two terms are monotone and piecewise constant, and the last term is monotone and piecewise linear. In particular, $\operatorname{vol}(v,r)$ is continuous, and on each open interval between consecutive vertex distances from $v$,
\[
\frac{d}{dr}\operatorname{vol}(v,r)=c\bigl(\delta(v,r)\bigr).
\]

\subsection{Cut procedure and its correctness}

Here we describe the modified region-growing rounding algorithm, which is almost identical to the original procedure in~\cite{FastApproximateGraphPartitioningAlgorithms}, except that it grows a region around an arbitrary terminal $v\in V'\cap X$ instead of an arbitrary vertex $v\in V'$. The procedure is given in Algorithm~\ref{alg:terminal-cut-proc}. It is used, whenever there exists at least one terminal in the current subgraph $G'=(V',E')$ and $w\br{V'}>\rho_0\cdot w\br{G}$. The procedure returns a sphere $U=N(v,r)$ for some radius $r\le \tilde r$ such that $w\br{U}\le \rho_0\cdot w\br{G}$ and the ratio between $c\br{\delta\br{v, r}}$ and $\operatorname{vol}(v,r)$ is $\bigo\br{\log n}$. Pick $
\tilde r=\frac{\rho_0-\rho}{\rho_0}.$

\begin{algorithm}[H]
\caption{\textsc{Cut-Proc} for Terminal-Constrained Region Growing}
\label{alg:terminal-cut-proc}
\DontPrintSemicolon
\KwIn{A connected subgraph $G'=(V',E')$, edge capacities $c$, LP lengths $d$, terminal set $X\subseteq V'$, with $V'\cap X\neq\emptyset$ and $w(V')>\rho_0\cdot w(V)$}
\KwOut{A nontrivial set $U\subsetneq V'$ satisfying \eqref{eq:terminal-good-radius}}
Set $\tilde r\gets (\rho_0-\rho)/\rho_0$\;
Choose an arbitrary center $v\in V'\cap X$\;
Set $U\gets\{v\}$\;
Set $v_0\gets$ closest vertex to $U$ in $V'\setminus U$ with respect to $\dist_{d,V'}$\;
\While{$c\br{U, V'\setminus U} > \frac{1}{\tilde r}\cdot\ln\left(\frac{\operatorname{vol}(v,\tilde r)}{\operatorname{vol}(v,0)}\right)\cdot \operatorname{vol}\br{v,\dist_{d,V'}\br{v,v_0}}$}{
	$U\gets U\cup\{v_0\}$\;
	Set $v_0\gets$ closest vertex to $U$ in $V'\setminus U$ with respect to $\dist_{d,V'}$\;
}
\Return{$U$}\;
\end{algorithm}

\begin{theorem}
For every processed subgraph $G'=(V',E')$ with $V'\cap X\neq\emptyset$ and
$w\br{V'}>\rho_0\cdot w\br{G}$, procedure \textsc{Cut-Proc} returns
nonempty $U\subsetneq V'$ such that:
\begin{enumerate}
    \item $U\cap X\neq\emptyset$,
    \item $w\br{U}\le \rho_0\cdot w\br{G}$,
    \item The following inequality holds for the returned radius:
    \[
c\bigl(\delta(v,r)\bigr)
\le
\frac{1}{\tilde r}\cdot\ln\left(\frac{\operatorname{vol}(v,\tilde r)}{\operatorname{vol}(v,0)}\right)
\cdot \operatorname{vol}(v,r).
\label{eq:terminal-good-radius}
\]
\end{enumerate}
\end{theorem}

\begin{proof}
Fix the chosen terminal center $v\in V'\cap X$. Let
$v=z_0,z_1,\dots,z_{|V'|-1}$ be vertices sorted by nondecreasing
distance from $v$ in $G'$. For each $i$, denote
\[
I_i=\bigl(\dist_{V'}(v,z_i),\dist_{V'}(v,z_{i+1})\bigr).
\]
On each $I_i$, $\operatorname{vol}(v,r)$ is linear and
\[
\operatorname{vol}'(v,r)=c\bigl(\delta(v,r)\bigr)
\qquad\text{for }r\in I_i.
\]

\begin{lemma}\label{lem:cut-proc-log-derivative-point}
There exists $r_0\in (0,\tilde r]\cap\bigcup_i I_i$ such that
\[
\frac{\operatorname{vol}'(v,r_0)}{\operatorname{vol}(v,r_0)}
\le
\frac{1}{\tilde r}
\ln\left(\frac{\operatorname{vol}(v,\tilde r)}{\operatorname{vol}(v,0)}\right).
\]
\end{lemma}
\begin{proof}
We use an averaging argument for the logarithmic derivative over $[0,\tilde r]$.
Assume to the contrary that for all admissible $r\in(0,\tilde r]$,
\[
\frac{\operatorname{vol}'(v,r)}{\operatorname{vol}(v,r)}
>
\frac{1}{\tilde r}\cdot
\ln\left(\frac{\operatorname{vol}(v,\tilde r)}{\operatorname{vol}(v,0)}\right).
\]
Integrating over $[0,\tilde r]$ gives
\[
\int_0^{\tilde r}\cdot\frac{\operatorname{vol}'(v,r)}{\operatorname{vol}(v,r)}dr
>
\frac{1}{\tilde r}\cdot
\ln\left(\frac{\operatorname{vol}(v,\tilde r)}{\operatorname{vol}(v,0)}\right)
\int_0^{\tilde r}dr.
\]
However, it is easy to see that both sides equal:
\[
\ln\operatorname{vol}(v,\tilde r)-\ln\operatorname{vol}(v,0)
=
\ln\left(\frac{\operatorname{vol}(v,\tilde r)}{\operatorname{vol}(v,0)}\right),
\]
which is a contradiction.
\end{proof}

\begin{lemma}\label{lem:cut-proc-scanned-radius}
There exists a scanned radius $r\le \tilde r$ satisfying \eqref{eq:terminal-good-radius}.
\end{lemma}
\begin{proof}
We first find a good continuous radius and then move to the nearest scanned radius.
Let $r_0$ be given by Lemma~\ref{lem:cut-proc-log-derivative-point}, and define
\[
r_1=\min\{\dist_{V'}(v,u)\colon u\in V',\ \dist_{V'}(v,u)\ge r_0\}.
\]
Then $N(v,r_0)=N(v,r_1)$, hence both radii induce the same cut and
\[
c\bigl(\delta(v,r_1)\bigr)=\operatorname{vol}'(v,r_0).
\]
By monotonicity $\operatorname{vol}(v,r_0)\le \operatorname{vol}(v,r_1)$, so \eqref{eq:terminal-good-radius} holds for $r_1$. Since the procedure scans radii exhaustively, it returns some scanned radius $r\le r_1\le \tilde r$ satisfying \eqref{eq:terminal-good-radius}.
\end{proof}

\begin{lemma}\label{lem:cut-proc-weight-bound}
For the returned set $U=N(v,r)$, we have $w\br{U}\le \rho_0\cdot w\br{G}$.
\end{lemma}
\begin{proof}
We prove this by contradiction: if the returned ball were too heavy, the radius lower bound for terminal-centered sets would force a vertex farther than $\tilde r$ from $v$ inside $U$.
If $w\br{U}>\rho_0\cdot w\br{G}$, then by terminal radius guarantee applied to $S=U$ and center $v\in U\cap X$,
\[
\operatorname{radius}(v,U)>\tilde r.
\]
But $U=N(v,r)$ and $r\le \tilde r$, so every $u\in U$ satisfies $\dist_{V'}(v,u)<r\le \tilde r$, contradiction.
Hence $w\br{U}\le \rho_0\cdot w\br{G}$.
\end{proof}

By Lemma~\ref{lem:cut-proc-scanned-radius}, property (3) of the theorem holds. By Lemma~\ref{lem:cut-proc-weight-bound}, property (2) holds. Property (1) is immediate since $v\in U\cap X$.
\end{proof}

\subsection{Global algorithm and approximation factor}

Run \textsc{Cut-Proc} iteratively as in Algorithm~\ref{alg:terminal-global-procedure}.

\begin{algorithm}[H]
\caption{Global Terminal Separator via Iterated \textsc{Cut-Proc}}
\label{alg:terminal-global-procedure}
\DontPrintSemicolon
\KwIn{A weighted graph $G=(V,E,c,w)$, terminal set $X\subseteq V$, and parameters $0<\rho<\rho_0<1$}
\KwOut{A feasible $(\rho_0,X)$-separator $F\subseteq E$}
Initialize $F\gets\emptyset$\;
\While{there exists a connected component $C$ of $(V,E\setminus F)$ with $C\cap X\neq\emptyset$ and $w(C)>\rho_0\cdot w(V)$}{
    Run Algorithm~\ref{alg:terminal-cut-proc} on $G[C]$ with terminal set $X\cap C$ and obtain $U\subsetneq C$\;
    Add boundary edges of $U$ in $G[C]$ to $F$, i.e.,
    $F\gets F\cup \delta_{G[C]}(U).$
}
\Return{$F$}\;
\end{algorithm}

At termination, every terminal-containing component has weight at most $\rho_0\cdot w\br{G}$, so the produced cut is a feasible $(\rho_0,X)$-separator. 
Let the procedure be called $q$ times and let $U_1,\dots,U_q$ be returned spheres.
By construction, these sets are pairwise disjoint.

For each call, by \eqref{eq:terminal-good-radius},
\[
c\bigl(\delta(U_j)\bigr)
\le
\frac{1}{\tilde r}\cdot
\ln\left(\frac{\operatorname{vol}_j(v_j,\tilde r)}{\operatorname{vol}_j(v_j,0)}\right)\cdot
\operatorname{vol}_j(v_j,r_j),
\]
where $\operatorname{vol}_j$ is the volume function within the current processed component.
Since
\[
\operatorname{vol}_j(v_j,\tilde r)\le (1+\epsilon)\cdot\tau_X,
\qquad
\operatorname{vol}_j(v_j,0)=\epsilon\cdot\tau_X,
\]
we get
\[
c\bigl(\delta(U_j)\bigr)
\le
\frac{1}{\tilde r}\cdot\ln\left(\frac{1+\epsilon}{\epsilon}\right)\cdot
\operatorname{vol}_j(v_j,r_j).
\]

Summing over calls, edge-length contributions are at most $\tau_X$, and seed terms contribute at most $q\cdot\epsilon\cdot\tau_X\le n\epsilon\cdot\tau_X$. Hence
\[
\sum_{j=1}^q c\bigl(\delta(U_j)\bigr)
\le
\frac{1}{\tilde r}\cdot 
\ln\left(\frac{1+\epsilon}{\epsilon}\right)
\cdot \tau_X(1+\epsilon n).
\]

Choosing
\[
\epsilon=\frac{1}{n\cdot\ln n},
\]
we obtain the following bound.

\begin{theorem}
The above algorithm computes a $(\rho_0,X)$-separator of capacity at most
\[
\left(\frac{\rho_0}{\rho_0-\rho}+o(1)\right)\ln n\cdot \tau_X.
\]
\end{theorem}

\end{document}